\documentclass[titlepage,11pt]{article}
\usepackage{graphicx} 
\usepackage{paralist}
 
\usepackage[a4paper,left=3cm, right=3cm,top=3cm, bottom=3cm]{geometry}
\usepackage{hyperref}
\usepackage{setspace}
\usepackage{color,soul}
\usepackage{url}
\usepackage[toc,page]{appendix}
\usepackage{amssymb}
\usepackage[nottoc]{tocbibind}
\usepackage{overpic}
\usepackage{amsmath}
\usepackage{amsthm}
\usepackage{tikz}
\usepackage{verbatim}
\usepackage{fbox}
\newtheorem{defi}{Definition}[section]
\theoremstyle{remark}
\newtheorem*{rem}{Remark}
\theoremstyle{plain}
\newcommand\ddfrac[2]{\frac{\displaystyle #1}{\displaystyle #2}}
\newtheorem{thm}[defi]{Theorem}
\newtheorem{prop}[defi]{Proposition}

\newtheorem{cor}[defi]{Corollary}
\newtheorem{lemma}[defi]{Lemma}
\usepackage{wrapfig}
\usepackage{lipsum}
\usepackage{etoolbox}
\renewcommand{\theta}{\vartheta}

\DeclareMathOperator{\cn}{cn}
\DeclareMathOperator{\sn}{sn}
\DeclareMathOperator{\ns}{ns}
\DeclareMathOperator{\cs}{cs}
\DeclareMathOperator{\dn}{dn}
\DeclareMathOperator{\sgn}{sgn}
\usepackage{makecell}
\begin{document}
\begin{titlepage} 
\vspace{-4cm}
		\begin{centering}
			
			\begin{figure}[!h]

				\begin{minipage}{0.25\linewidth}
					\begin{center}
						\includegraphics[scale=.5]{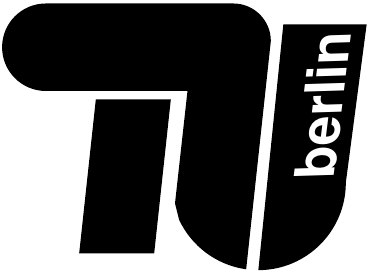} 
					\end{center}  
				\end{minipage}
				\begin{minipage}{0.4\linewidth}
					\text{Technische Universit\"at Berlin }\\
					\text{Institut f\"ur Mathematik}
				\end{minipage}
				\hfill
			\end{figure}
			
			\LARGE
			
			Masterarbeit\\[.7cm]
			
			\textbf{Geometry of the discrete time Euler top and related 3-dimensional birational maps}\\[.7cm]

			\large


					\includegraphics[width = .8\linewidth]{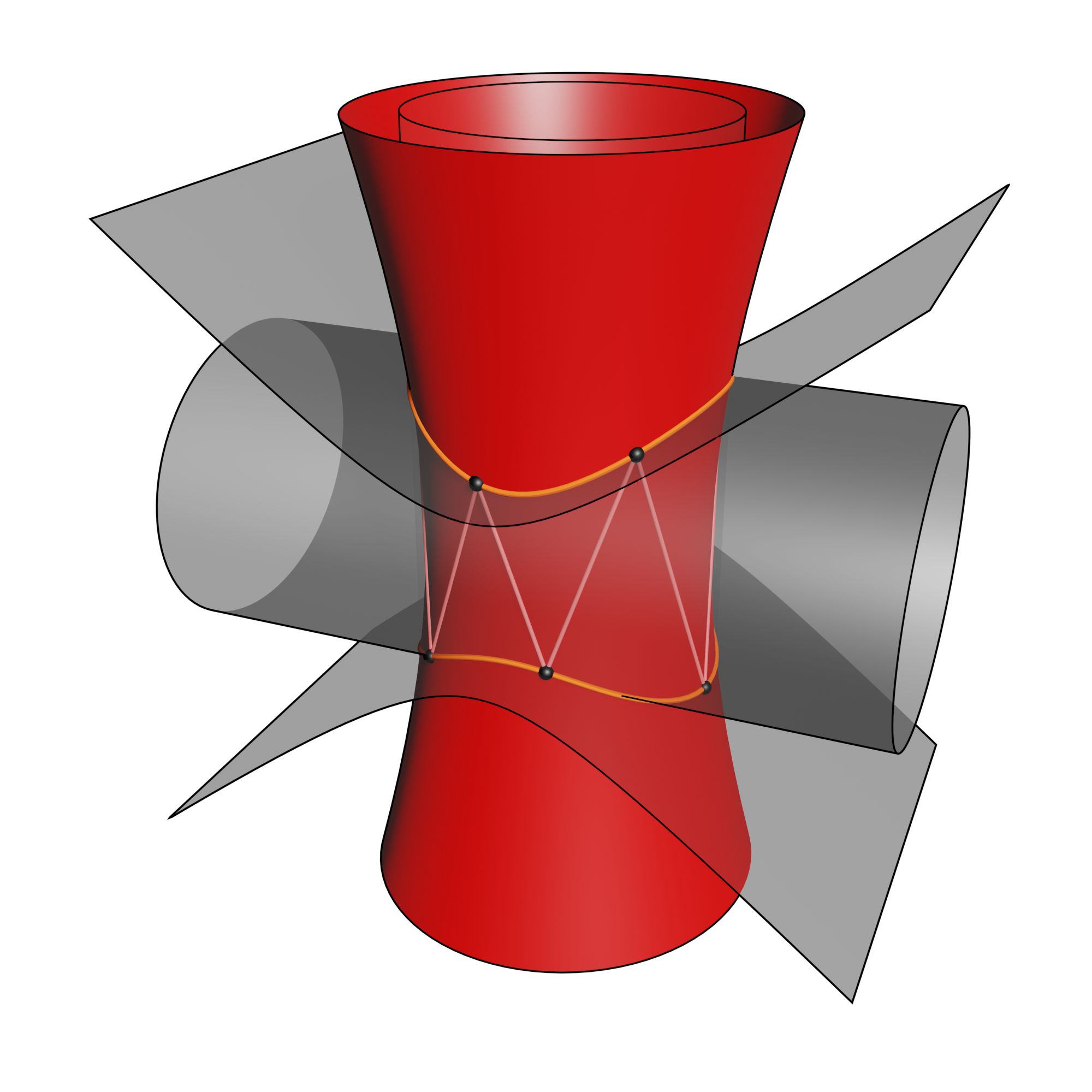}\\
			Nina Smeenk, Matrikelnummer 366480\\
			Berlin, den 27. November 2020\\
			[.7cm]
			\begin{minipage}{\linewidth} 
				\begin{tabbing}
					Erstgutachter:\quad \=  Prof. Dr. Yuri B. Suris, TU Berlin\\
					Zweitgutachter:\quad \= Prof. Dr. Alexander I. Bobenko, TU Berlin \\
				\end{tabbing}
			\end{minipage}

		\end{centering}
	\end{titlepage}
	
\newpage\null\thispagestyle{empty}    
\setcounter{page}{2}



\newpage\null\thispagestyle{empty}
\section*{Zusammenfassung}
Der Euler Kreisel ist ein bedeutendes und umfangreich erforschtes integrables dynamisches System. Diskretisierungen, welche die Integrabilit\"at des stetigen Systems erhalten, wurden erforscht und viele weisen, neben der Integrabilit\"at, weitere bedeutende Eigenschaften auf. Mit einer dieser integrablen Diskretisierungen nach Hirota und Kimura werden wir uns geometrisch auseinandersetzten. Die expliziten L\"osungen dieser Diskretisierung k\"onnen mit Hilfe einer birationalen Abbildung oder aber mit Hilfe Jacobi elliptischer Funktionen beschrieben werden. 
Die zugeh\"origen Erhaltungsgr\"o{\ss}en k\"onnen als Zylinder im 3-dimensionalen reellen Raum dargestellt werden. Die L\"osungen des Systems beschr\"anken sich dann auf die Schnittkurve der Zylinder. Mit Hilfe von Quadriken, in der von den Zylindern aufgesapnnten Schar, werden wir Involutionen definieren, sodass wir den diskreten Euler Kreisel als Komposition zweier dieser Involutionen beschreiben k\"onnen. Einerseits betrachten wir diese Abbildungen als komplexe Involutionen, beschrieben durch Jacobi elliptische Funkionen, andererseits als reelle Involutionen. Fallen die beiden Involutionen im reellen Fall zusammen, wird gezeigt, dass es sich hierbei (bis auf Vorzeichen) um die Wurzel des diskreten Euler Kreisels handelt. 
Au{\ss}erdem wird untersucht, wann die reellen Involutionen rational sind, d.h. wann es sich um 3-dimensionale, birationale Abbildungen handelt.

\newpage\newpage\null\thispagestyle{empty}\newpage\tableofcontents
\newpage\null\thispagestyle{empty}\newpage
\pagenumbering{arabic}
\vspace{.3cm}
\section{Introduction}
\vspace{.4cm}
As one of the most famous integrable systems the Euler top and its integrable discretizations have been studied by many authors. Despite the great variety of known results, these systems continue to show remarkable new features when investigating them in various new settings.\\ \\
In this thesis we will study the geometry of an explicit discretization of the continuous equations of motions of the Euler top introduced by Hirota and Kimura in \cite{hirota2000discretization}, denoted by the HK-type discretization of the Euler top.
Among other significant features it has been shown that this system is integrable with three conserved quantities, two of which are independent.
These conserved quantities describe three cylinders in 3-dimensional real space. Any two such cylinders define a pencil consisting of quadrics defined by linear combination of the given two. Since only two of the cylinders are independent the third one is also a member of the pencil. All quadrics in a pencil intersect in the same 3-dimensional space curve, hence also the three cylinders. This specific intersection curve consists of two real connected components. Determined by the orientation of the conserved quantities, there are two, projectively equivalent, cases that are shown in Figure \ref{Fig:case_a_b}. \\ \\
We are interested in maps that interchange points on the two components of the intersection curve.
Geometrically these maps can be described in the following way: We fix a ruled quadric in the pencil given by the conserved quantities and a point on the base curve of the pencil, i.e., the intersection curve. The rulings of the quadric through the given point intersect the base curve in exactly one other point, lying on the other component of the intersection curve. The maps will be defined by interchanging these two points.\\ \\
After introducing the discrete system and its properties we will review how to parameterize the two real components of the intersection curve in terms of Jacobi elliptic functions. By discretizing this parametrization with a fixed time step, that can be found from the initial parameters of the system, the explicit solutions of the discrete system can also be stated in terms of Jacobi elliptic functions.
Using this description of the solutions we can find two ruled quadrics such that the composition of the previously described geometric maps, that interchange points on the intersection curve components along generators of the quadrics, matches the solutions of the HK-type discretization of the Euler top. The Figure on the title page shows iterations of these maps for the case where the two quadrics coincide.\newpage
\noindent Furthermore we extend the parametrization of the real components of the intersection curve to the complex numbers. In this case the intersection curve constitutes an embedding of a torus in $\mathbb{C}^3$ and we find that the geometric maps can explicitly be described as complex involutions. 
Therefore we find that the birational map describing the given discretization of the Euler top can be expressed as the composition of two complex involutions in terms of Jacobi elliptic functions. \\ \\
The same geometric maps can also be described with real parameters omitting Jacobi elliptic functions. By  first stating them for a general one-sheeted hyperboloid and a cylinder we already see that these maps are real involutions.\\ \\
We then specify these involutions to the cylinders given by the conserved quantities and ruled quadrics in their pencil.
Therefore we also find that the birational map describing the solutions of the given discretization of the Euler top can expressed as the composition of two real involutions.\\ \\ In the remaining part we will first investigate the case where the two involutions coincide. This case is in close correspondence to the square root of the discrete time Euler top. Subsequently we will give explicit examples where the involutions turn into birational 3-dimensional real maps.\\ \\ The visualizations in the thesis are done in Blender 2.8 using the python ddg library 'pyddglab' created at the Geometry Group at TU Berlin.
\newpage

\section{Geometry of the discrete time Euler top}
\vspace{.4cm}
As one of the most famous examples of integrable systems the Euler top and its integrable discretizations have been studied extensively, see for example \cite{adler2013algebraic, audin1999spinning}. Spinning tops, including the Euler top, are dynamical systems that describe the motions of rigid bodies with a fixed point under the influence of a constant gravity field. In general these systems are not integrable, however it has been shown that exactly three holonomic integrable tops, namely the Euler, Lagrange and Kovalevskaya top, exist, see \cite[pp 419-468]{adler2013algebraic}. These tops can be distinguished by their symmetry and position of their center of gravity. The Euler top does not admit any symmetry in its principal axis and its fixed point coincides with its center of gravity. The continuous equations of motion read 
\begin{align}
\label{contin_eq_of_motion}
\begin{cases} 
\dot{x}_1 = \alpha_1x_2x_3\\
\dot{x}_2 = \alpha_2x_1x_3\\
\dot{x}_3 = \alpha_3x_1x_2 
\end{cases}
\end{align}
with real parameters $\alpha_i$ encoding the moments of inertia in the body frame.  Even though many results are known, new results can be shown when investigating the system, as well as its integrable discretizations, in various new settings. Therefore we will take a closer look at the geometric background of an explicit discretization of the equations
of motion of the Euler top introduced by Hirota and Kimura in \cite{hirota2000discretization}. This discrete system is given by 
\begin{align}
\begin{cases} \label{equations_of_motion}
\tilde{x}_1 - x_1 =  \frac{\varepsilon \alpha_1}{2}(\tilde{x}_2x_3 + x_2\tilde{x}_3)\\
\tilde{x}_2 - x_2 =  \frac{\varepsilon \alpha_2}{2}(\tilde{x}_3x_1 + x_3\tilde{x}_1) \\
\tilde{x}_3 - x_3 = \frac{\varepsilon \alpha_3}{2}(\tilde{x}_1x_2 + x_1\tilde{x}_2)
\end{cases}
\end{align}
with real parameters $\alpha_i$ and $\epsilon$ where the latter denotes the discrete time step. For the parameters $\alpha_i$ we will always assume \begin{align} \label{sign_alphs}\alpha_1 < 0, \qquad \alpha_2 > 0, \qquad \alpha_3 < 0. \end{align}
For simplicity we introduce $$\delta_i := \frac{\epsilon \alpha_i}{2}.$$
Throughout this thesis when referring to the discrete time Euler top (dEt) the explicit discretization given by (\ref{equations_of_motion}) is meant.\newpage \noindent Explicitly the solutions of (\ref{equations_of_motion}) are given by the birational map \begin{align}
\label{explicit_f}
f(\cdot,\delta) \colon \mathbb{R}^3 \rightarrow \mathbb{R}^3, \quad f(x,\delta)=\tilde{x} \end{align}
with $\delta=(\delta_1,\delta_2,\delta_3)$ and 
{\setstretch{2.0}\begin{align*}
\begin{cases}\tilde{x}_1&=
\ddfrac{x_1+2\delta_1x_2x_3+x_1(-\delta_2\delta_3x_1^2+\delta_1\delta_3x_2^2+\delta_1\delta_2x_3^2)}{1-\delta_2\delta_3x_1^2-\delta_1\delta_3x_2^2-\delta_1\delta_2x_3^2 - 2\delta_1\delta_2\delta_3x_1x_2x_3}\\
\tilde{x}_2 &=
\ddfrac{x_2+2\delta_2x_1x_3+x_2(\delta_2\delta_3x_1^2-\delta_1\delta_3x_2^2+\delta_1\delta_2x_3^2)}{1-\delta_2\delta_3x_1^2-\delta_1\delta_3x_2^2-\delta_1\delta_2x_3^2 - 2\delta_1\delta_2\delta_3x_1x_2x_3}\\
\tilde{x}_3&= 
\ddfrac{x_3+2\delta_3x_1x_2+x_3(\delta_2\delta_3x_1^2+\delta_1\delta_3x_2^2-\delta_1\delta_2x_3^2)}{1-\delta_2\delta_3x_1^2-\delta_1\delta_3x_2^2-\delta_1\delta_2x_3^2 - 2\delta_1\delta_2\delta_3x_1x_2x_3}.
\end{cases}
\end{align*}}
\noindent \hspace{-.7em} The inverse $f^{-1}(\cdot,\delta)=f(\cdot,-\delta)$ can be found by inverting the signs of the $\delta_i$ and therefore yields the solutions of the dEt with negative time step $\epsilon$.\\ \\
This discretization of the Euler top has the conserved quantities \begin{align}
\label{F_i}F_1 = \frac{1-\delta_3\delta_1x_2^2}{1-\delta_1\delta_2x_3^2} \qquad F_2 = \frac{1-\delta_1\delta_2x_3^2}{1-\delta_2\delta_3x_1^2} \qquad F_3 = \frac{1-\delta_2\delta_3x_1^2}{1-\delta_3\delta_1x_2^2}
\end{align}
of which two are independent since $F_1F_2F_3=1$ (see \cite{pfadler2011bilinear}, also for more on their correspondence to the conserved quantities of the continuous case).
Geometrically these integrals describe the cylinders
\begin{equation}
\label{C_i}
\begin{aligned}
\mathcal{C}_1&\colon \quad \delta_1 \delta_3x_2^2 - F_1\delta_1 \delta_2x_3^2= 1-F_1\\
\mathcal{C}_2&\colon -F_2\delta_2 \delta_3x_1^2 + \delta_1 \delta_2x_3^2= 1-F_2\\ \text{ and } \
\mathcal{C}_3&\colon \quad \delta_2 \delta_3x_1^2 - F_3\delta_1 \delta_3x_2^2= 1-F_3. 
\end{aligned}
\end{equation}
Starting with a point $X = (X_1,X_2,X_3)\in \mathbb{R}^3$ and fixed parameter $\delta$, the values $F_1,F_2$ and $F_3$ and therefore the cylinders $\mathcal{C}_1,\mathcal{C}_2$ and $\mathcal{C}_3$ are uniquely defined. The linear combination of any two such cylinders define a pencil of quadrics. With the dependency of the third cylinder we find that it is included in the pencil. Since all quadrics in a pencil intersect in a common intersection curve, called the \emph{base curve} of the pencil, also $\mathcal{C}_1,\mathcal{C}_2$ and $\mathcal{C}_3$ have a common intersection curve. This curve consists of two real connected components. The solutions of the difference equations (\ref{equations_of_motion}) will lie on the common intersection curve of $\mathcal{C}_1,\mathcal{C}_2$ and $\mathcal{C}_3$. They even restrict to one of the components that is determined by the position of the initial condition $X$.\\ \\ By the previous assumption on the signs of $\alpha$ and hence $\delta$ in (\ref{sign_alphs}) we find, for small enough $\epsilon>0$, that $F_1 \in (0,1)$ and $F_3>1$, which implies that $\mathcal{C}_1$ and $\mathcal{C}_3$ are elliptic cylinders. \\ \\The geometry of the system splits into two projectively equivalent cases:  
\begin{enumerate}
\item[case (a):] Let $F_2>1$ or equivalently $ F_1 < F_3^{-1} < 1$. Then $F_2$ defines a hyperbolic cylinder of two sheets where the sheets are separated by the x,y-plane. Hence the two components of the intersection curve are mirrored on the x,y-plane. For a given parametrization of one component, negation of the z-coordinate yields a parametrization of the other component, see Figure \ref{Fig:case_a_b} (left).
\item[case (b):]Let $F_2<1$ or equivalently $ F_3^{-1}< F_1 < 1$. Then $F_2$ defines a hyperbolic cylinder of two sheets where the sheets are separated by the y,z-plane. Hence the two components of the intersection curve are mirrored on the y,z-plane. For a given parametrization of one component negation of the x-coordinate yields a parametrization of the other component, see Figure \ref{Fig:case_a_b} (right).
\end{enumerate}
For both cases we will now review the parametrization of the intersection curve in terms of elliptic functions, the corresponding discretization that yields the solutions of the dEt and introduce hyperboloids in the pencil of $\mathcal{C}_1,\mathcal{C}_2$ and $\mathcal{C}_3$.
On these hyperboloids we will geometrically find two maps such that their composition matches the solutions of the dEt. For this let $\sigma \in \{a, b\}$ denote any of the two given cases.
\begin{figure}[t]
\begin{minipage}{0.53\linewidth}
\begin{overpic}[width=.9\linewidth]
{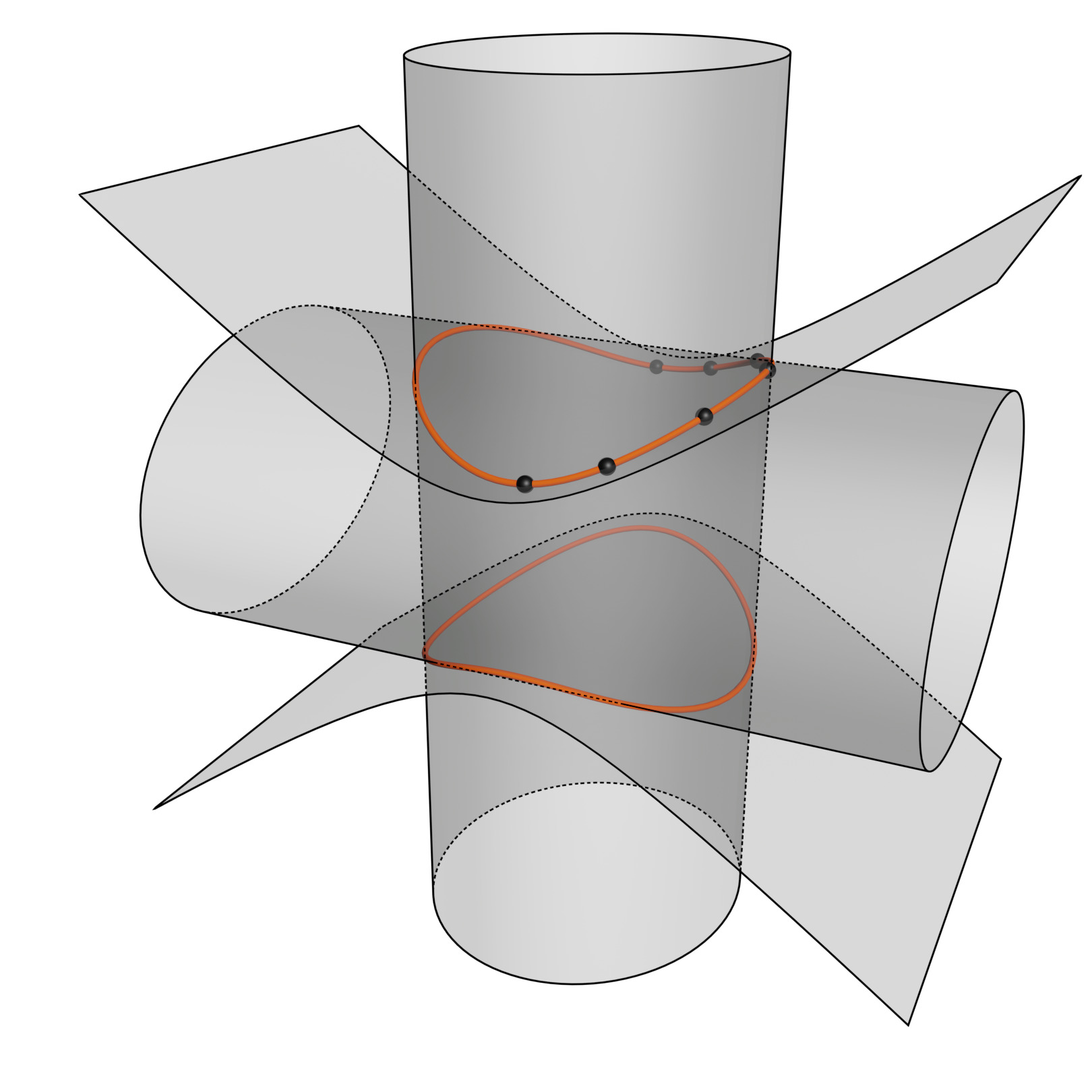}
\put(46,57){$\scriptstyle X$}
\end{overpic}
\end{minipage}
\hfill
\begin{minipage}{0.48\linewidth}
\begin{overpic}[width=\linewidth]
{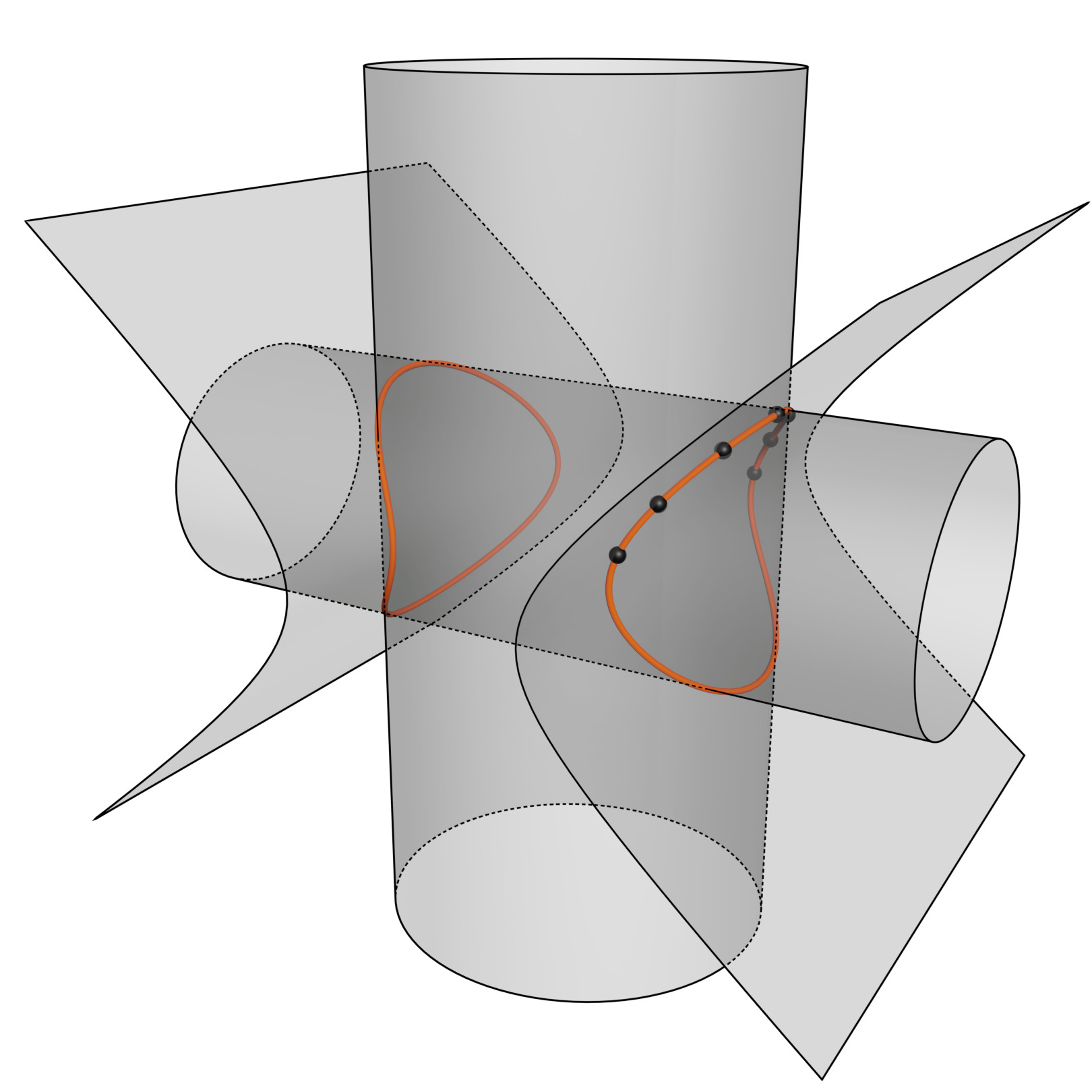} 
\put(58,46){$\scriptstyle X$}
\end{overpic}
\end{minipage}
\caption{The HK-type discretization of the Euler top: The quadrics symbolize the conserved quantities of the discrete systems and solutions of the system lie on their common intersection curve. The left picture corresponds to case (a) and the right one to case (b).}
\label{Fig:case_a_b}
\end{figure}

\subsection{Explicit solutions of the discrete time Euler top in terms of elliptic functions}
\vspace{.4cm}
\label{Sec:case_a}
Explicit solutions of the discrete system (\ref{equations_of_motion}) in terms of elliptic functions have been investigated in \cite{kimura2000discretization} and extended in 
\cite{petrera2010hamiltonian}.
This section will provide a summary of the results shown in both. We will use the Jacobi elliptic functions $\sn(u,k)$, $\cn(u,k)$ and $\dn(u,k)$ with the argument $u \in \mathbb{C}$ and the parameter $k \in [0,1]$. These functions generalize the trigonometric functions $\sin$ and $\cos$ for a circle to trigonometric functions for conic sections. The parameter $k$ is called the \emph{modulus}. Beside the definition of $k$, definitions of $\sn(u,k)$, $\cn(u,k)$ and $\dn(u,k)$ in terms of theta functions can be found in \cite{Nist}. Additionally, \cite{Nist} contains a great overview of elementary identities, addition theorems and much more, which we will often refer to. For fixed modulus $k$ the Jacobi elliptic functions, as functions in $u$, are doubly periodic, meromorphic functions with simple poles and zeros. For $u\in\mathbb{R}$ the functions take on real values. Since we are interested in parameterizing the intersection curve in 3-dimensional real space we generally restrict to $u\in \mathbb{R}$, except in Section \ref{Sec:complex}, where we will extend to complex arguments.
Throughout the thesis we will omit the modulus $k$ as a second parameter.
The moduli for the cases (a) and (b) will be given in (\ref{case_a_A_k}) and (\ref{case_b_B_k}) and stay fixed throughout the rest of the thesis. \\ \\
In case (a) we can parameterize the intersection curve of the cylinders $\mathcal{C}_1,\mathcal{C}_2$ and $\mathcal{C}_3$ by 
\begin{equation*}
v^{a}_\pm\colon \mathbb{R}\rightarrow \mathbb{R}^3, \quad u \mapsto \begin{pmatrix}
A_1 \cn(u)\\
\pm A_2 \sn(u)\\
\pm A_3 \dn(u)
\end{pmatrix}
\end{equation*}
and in case (b) by 
\begin{equation*}
v^{b}_\pm\colon \mathbb{R}\rightarrow \mathbb{R}^3, \quad u \mapsto \begin{pmatrix}
\pm B_1 \dn(u)\\
\pm B_2 \sn(u)\\
B_3 \cn(u)
\end{pmatrix}
\end{equation*}
where $A_i, B_i \in\mathbb{R}$ and the moduli $k_a,k_b\in [0,1]$ are to be determined. The sign $\pm$ distinguishes the two real components of the intersection curve and the orientation of their parametrization. The given formulas describe the two real components of the intersection curve with opposite orientations. To determine the parameters $A_i$ and $k_a$, we find that substitution into $\mathcal{C}_3$ yields
$$\delta_2 \delta_3A_1^2 \cn^2(u) - F_3\delta_3 \delta_1A_2^2\sn^2(u)= 1-F_3$$
and into $\mathcal{C}_1$:
\begin{align*}
\delta_3 \delta_1 A_2^2 \sn^2(u) - F_1\delta_1 \delta_2 A_3^2\dn^2(u) &= 1-F_1.
\end{align*}
Using the basic identities $$\sn^2(u)+\cn^2(u)=k^2\sn^2(u)+\dn^2(u) = 1$$ for Jacobi elliptic functions with common modulus, by the previous substitution, we derive at
\begin{equation}
\label{case_a_A_k}
\begin{aligned}
A_1^2 &= \frac{1-F_3}{\delta_2 \delta_3}, & A_2^2 &= \frac{1-F_3^{-1}}{\delta_3 \delta_1},\\
A_3^2 &= \frac{1-F_1^{-1}}{\delta_1 \delta_2}, & k_{a}^2 &= \frac{1-F_3^{-1}}{1- F_1}.
\end{aligned}
\end{equation}
Similar substitution for case (b) yields 
\begin{equation}
\label{case_b_B_k}
\begin{aligned}
B_2^2 &= \frac{1-F_1}{\delta_1 \delta_3}, & B_3^2 &= \frac{1-F_1^{-1}}{\delta_1 \delta_2},\\
B_1^2 &= \frac{1-F_3}{\delta_2 \delta_3}, &  k_{b}^2 &= \frac{1-F_1}{1- F_3^{-1}}.
\end{aligned}
\end{equation}
Thus we derive at the full parametrization
\begin{align}
\label{case_a_v}
v^{a}_\pm\colon [0,4K_{a})\rightarrow \mathbb{R}^3, \quad u \mapsto \begin{pmatrix}
A_1 \cn(u)\\
\pm A_2 \sn(u)\\
\pm A_3 \dn(u)
\end{pmatrix}
\end{align}
respectively 
\begin{align}
\label{case_b_v}
v^{b}_\pm\colon [0,4K_{b})\rightarrow \mathbb{R}^3, \quad u \mapsto \begin{pmatrix}
\pm B_1 \dn(u)\\
\pm B_2 \sn(u)\\
B_3 \cn(u)
\end{pmatrix}
\end{align}
with $K_{\sigma} = \mathcal{K}(k_{\sigma})$ being the real quarter period where $\mathcal{K}(\cdot)$ denotes the elliptic integral of the first kind. For a formal definition, see \cite[§19.2(ii)]{Nist}.
A uniform sampling yields the following discretizations:$$v^{a}_\pm\colon \mathbb{Z} \rightarrow \mathbb{R}^3, \qquad n \mapsto \begin{pmatrix}
A_1 \cn(u_0 + \nu n)\\
\pm A_2 \sn(u_0 + \nu n )\\
\pm A_3 \dn(u_0 + \nu n)
\end{pmatrix}$$
respectively 
 $$v^{b}_\pm\colon \mathbb{Z} \rightarrow \mathbb{R}^3, \qquad n \mapsto \begin{pmatrix}
\pm  B_1 \dn(u_0 + \nu n)\\
\pm B_2 \sn(u_0 + \nu n )\\
B_3 \cn(u_0 + \nu n)\\
\end{pmatrix}$$
for fixed start value $u_0$ and time step $\nu$. Note that this becomes periodic iff $ N \nu = 4K_{\sigma}-u_0$ for some $N \in \mathbb{Z}$.\\ \\
Now let $X\in \mathbb{R}^3$. For given parameter $\delta$ of the dEt the cylinders $\mathcal{C}_1, \mathcal{C}_2$ and $\mathcal{C}_3$ are uniquely defined by $X$. The value of $F_2$ then determines the case $\sigma \in \{a,b\}$. Depending on $\sigma$, (\ref{case_a_A_k}) or (\ref{case_b_B_k}) determines the coefficients of the parametrization of the intersection curve of $\mathcal{C}_1, \mathcal{C}_2$ and $\mathcal{C}_3$, up to sign. By construction $X$ lies on this intersection curve such that $X$ can be expressed as  \begin{align*} 
X =v^{a}(u_0):= \begin{pmatrix}
A_1\cn(u_0) \\
A_2\sn(u_0) \\
A_3\dn(u_0)
\end{pmatrix}\quad 
\text{or} \quad 
X =v^{b}(u_0):= \begin{pmatrix}
B_1\dn(u_0) \\
B_2\sn(u_0) \\
B_3\cn(u_0)
\end{pmatrix}
\end{align*}
with $u_0\in[-2K_{\sigma},2K_{\sigma}]$. Here we omit the $\pm$ in the notation of $v$ since the position of $X$ determines exactly one of the components of the intersection curve which will be of interest now. The signs of $A_1,A_3$ resp.\ $B_1,B_3$ are uniquely determined by the values of $X$. The sign of $A_2$ resp.\ $B_2$ can be switched by changing $u_0$ to $-u_0$, which corresponds to reversing the orientation of the component. Note that this does not affect the remaining two components since $\cn$ and $\dn$ are even functions.
We now wish to determine the parameter $\nu$ such that the previous discretization yields the solution of the dEt. We will denote this shift as the \emph{elliptic time step} of the dEt. \\ \\ We rewrite the starting point as \mbox{$X= v^{\sigma}(u_0) = v^{\sigma}(\tilde{u}_0 - \frac{\nu}{2})$} and wish to find $\nu$ such that $\tilde{X}:=v^{\sigma}(\tilde{u}_0 +\frac{\nu}{2})$ fulfills the equations of motion (\ref{equations_of_motion}). With the addition theorems of the Jacobi elliptic functions stated in \cite[section 3.4]{petrera2010hamiltonian} it can be shown, that calculation of $v^{\sigma}(u_0 + \frac{\nu}{2}) - v^{\sigma}(u_0 -\frac{\nu}{2})$ yields the following identities for $\nu$:
\begin{align*}
A_1 \sn\left(\frac{\nu}{2}\right)\dn\left(\frac{\nu}{2}\right) &= -\delta_1 A_2 A_3 \cn\left(\frac{\nu}{2}\right)\\
A_2 \sn\left(\frac{\nu}{2}\right) &= \ \delta_2 A_1 A_3 \cn\left(\frac{\nu}{2}\right)\dn\left(\frac{\nu}{2}\right) \\
A_3\ k_{a}^2 \sn\left(\frac{\nu}{2}\right) \cn\left(\frac{\nu}{2}\right) &= - \delta_3 A_1 A_2 \dn\left(\frac{\nu}{2}\right)
\end{align*}
respectively
\begin{align*}
B_1\ k_{b}^2\sn(\frac{\nu}{2})\cn(\frac{\nu}{2})  &= -\delta_1 B_2 B_3 \dn(\frac{\nu}{2})\\
B_2 \sn(\frac{\nu}{2}) &= \ \delta_2 B_1 B_3 \cn(\frac{\nu}{2})\dn(\frac{\nu}{2}) \\
B_3 \sn(\frac{\nu}{2}) \dn(\frac{\nu}{2})  &= - \delta_3 B_1 B_2  \cn(\frac{\nu}{2}).
\end{align*}
In conclusion we find the following:
\begin{prop}
\label{Prop:nu} For given parameter $\delta$ of the dEt and initial condition\\ \mbox{$X := v_\pm^{\sigma}(u_0) \in \mathbb{R}^3$}  the solutions of the discrete equations of motion (\ref{equations_of_motion}) are given by  $$
 n \mapsto v_\pm^{\sigma}(u_0+n\nu) $$
with $n \in \mathbb{N}$ and $\nu>0$ implicitly given by 
\begin{align}
\label{case_a_nu}
&\sn^2\left(\frac{\nu}{2}\right) = 1-F_1 \qquad \text{for case (a)} \\
\text{and } & 
\label{case_b_nu}
\sn^2\left(\frac{\nu}{2}\right) = 1-F_3^{-1} \quad \text{for case (b).}
\end{align}
\end{prop}
\noindent For a fixed starting condition $X=v_\pm^a(u_0)$, the restriction of the elliptic time step $\nu$ to be positive, as in the previous Proposition, will determine the sign of $A_2$ uniquely and thus will fix an orientation of the component of the intersection curve that includes $X$. Let $\hat{X} = v^a_\mp(-u_0)$. By definition of the parametrization we have $\hat{X}=(X_1,X_2,-X_3)$, if  $X = (X_1,X_2,X_3)$. The behavior of the discrete equations of motion with starting condition $\hat{X}$ can be investigated by direct computation: $$f(\hat{X}, \delta ) = \begin{pmatrix}
f(X,-\delta)_1\\
f(X,-\delta)_2\\
-f(X,-\delta)_3
\end{pmatrix}.$$
Hence for starting points on different components of the same intersection curve, the dEt evolves in different parametrization directions, i.e., it induces a parametrization such that the two components are reversely oriented, see Figure \ref{Fig:dEt_reversed_orientation}. This can easily also be verifired for case (b). The definition of the parametrization of the components $v^\sigma_+$ and $v_-^\sigma$ mirrors this behavior. \\ \\
\begin{figure}[h]
\begin{minipage}{0.53\linewidth}
\begin{overpic}[width=.9\linewidth]
{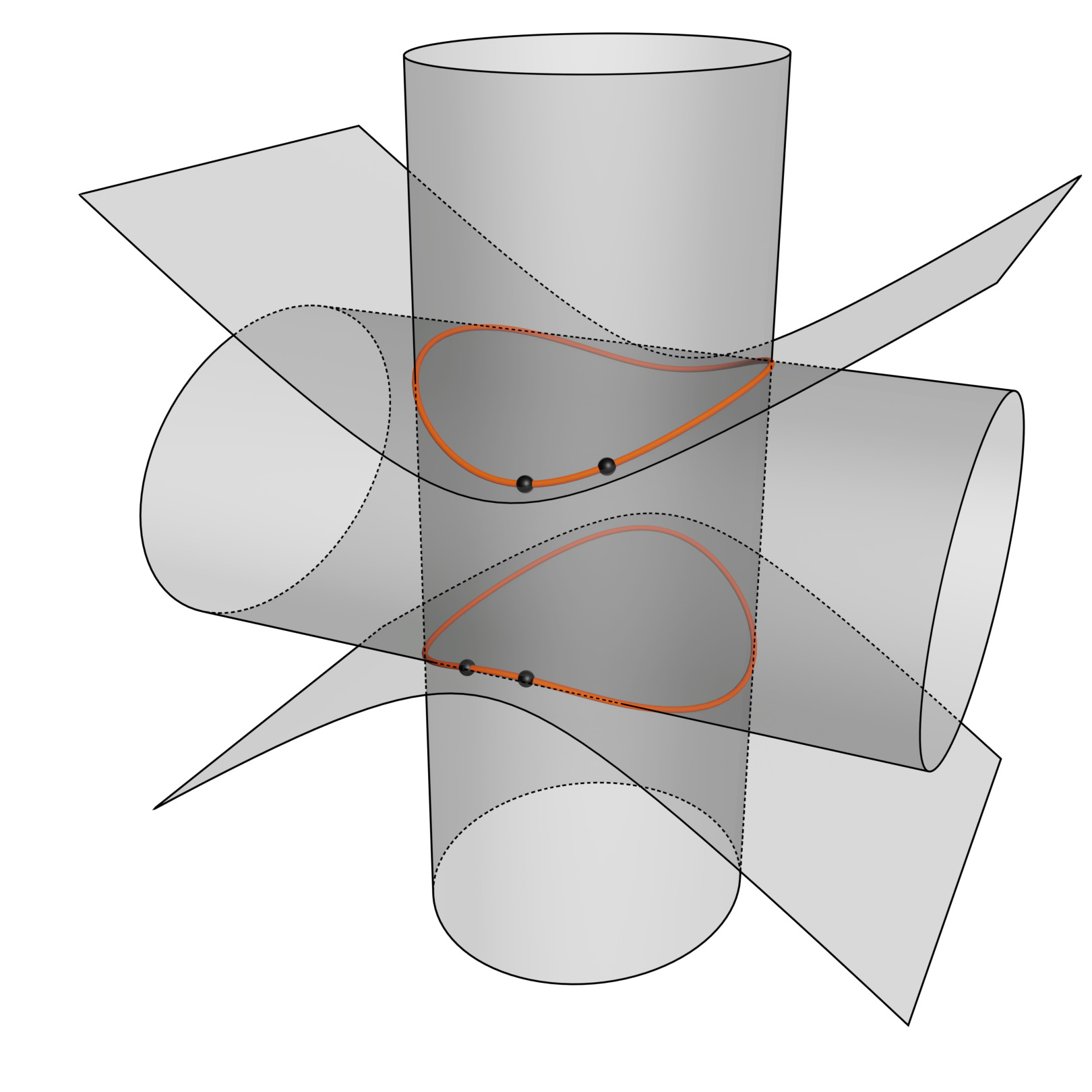}
\put(47,57){$\scriptstyle X$}
\put(47,39){$\scriptstyle \hat{X}$}
\end{overpic}
\end{minipage}
\hfill
\begin{minipage}{0.48\linewidth}
\begin{overpic}[width=\linewidth]
{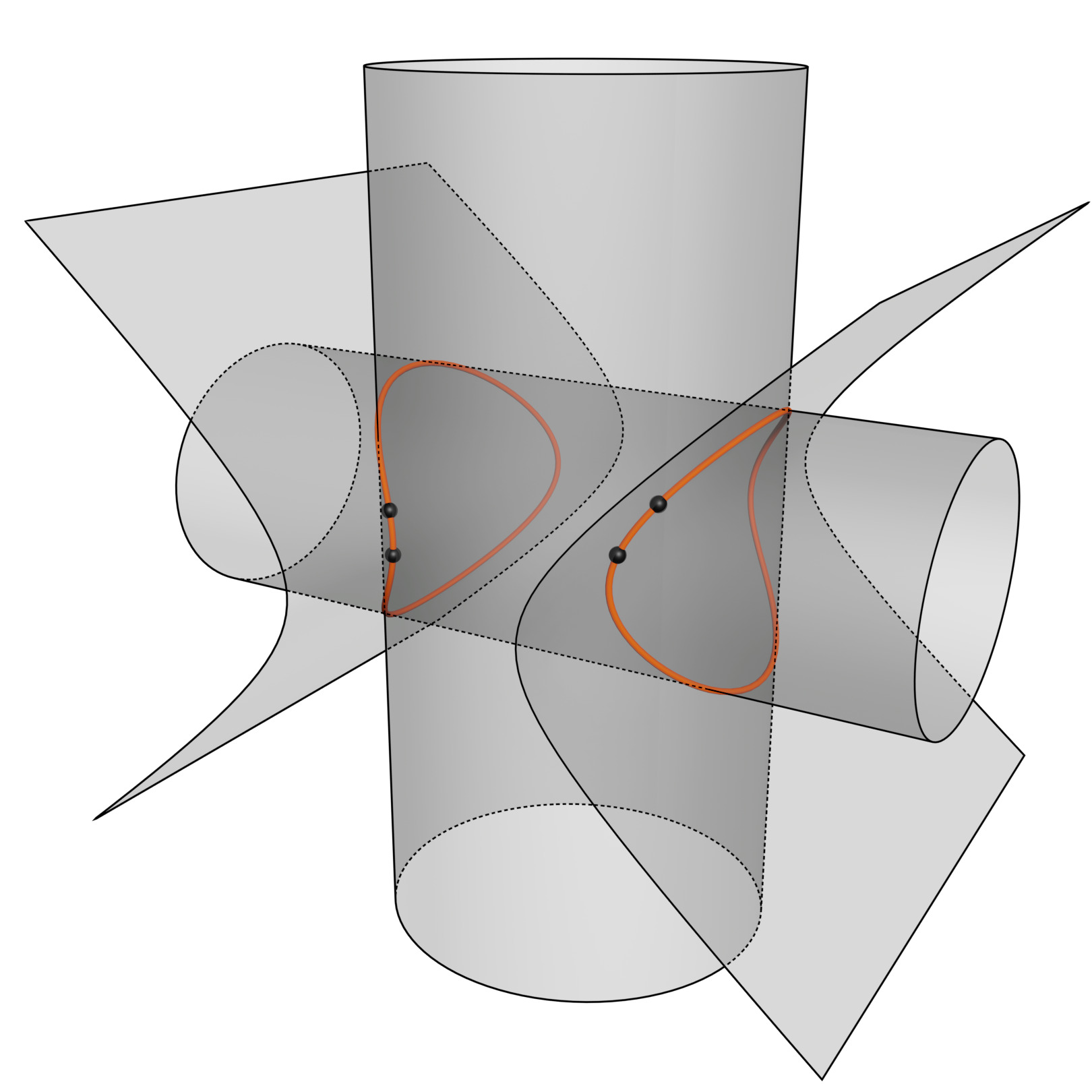} 
\put(58,48){$\scriptstyle X$}
\put(37,52){$\scriptstyle \hat{X}$}
\end{overpic}
\end{minipage}
\caption{Two starting conditions $X$ and $\hat{X}$ on two different components and their image under the birational map $f(\cdot,\delta)$ which restricts to the initial component that is given by the starting condition. The left picture corresponds to case (a) and the right one to case (b).}
\label{Fig:dEt_reversed_orientation}
\end{figure}
\vspace{.4cm}
\subsection{The discrete time Euler top as the composition of two maps acting on ruled surfaces}
\vspace{.4cm}
\label{Sec:geometric_maps}
Let $\nu>0$ be the elliptic time step of the dEt given by (\ref{case_a_nu}) or (\ref{case_b_nu}). We introduce $\nu_1$ and $\nu_2$ such that \begin{align}
\label{nu_i_def}\nu_1+\nu_2=\nu \quad \text{ with} \quad \nu_i \in I := [-2K_{\sigma},2K_{\sigma}]
\end{align} and with this the  maps
\begin{align}
\label{tau1}
\tau^{\sigma}_{\nu_1}&\colon v^{\sigma}_\pm(u_0) \mapsto v^{\sigma}_\mp(-u_0 -\nu_1)\\
\label{tau2}
\tau^{\sigma}_{\nu_2}&\colon v^{\sigma}_\mp(u_0) \mapsto v^{\sigma}_\pm(-u_0 +\nu_2).
\end{align} 
Geometrically they describe walking along the generators of ruled quadrics of the pencil spanned by $\mathcal{C}_1, \mathcal{C}_2$ and $\mathcal{C}_3$ from one component of the intersection curve to the other. The ruled quadric, denoted by $\mathcal{H}_i$, is determined by the value of $\nu_i$.\\ \\
With $\nu$ being the elliptic time step and
 $$(\tau^{\sigma}_{\nu_2}\circ\tau^{\sigma}_{\nu_1})(v_\pm^\sigma(u_0)) = v_\pm^\sigma(u_0+\nu),$$ it is easily verified, that the composition $\tau^{\sigma}_{\nu_2}\circ\tau^{\sigma}_{\nu_1}$ matches the solution of the dEt, see Figure \ref{Fig:tau}. Still missing is the explicit description of the ruled quadric $\mathcal{H}_i := \mathcal{C}_1+\lambda_i\mathcal{C}_3$, that includes the line from \mbox{ $v^{\sigma}_\pm(u_0)$ to  $v^{\sigma}_\mp(-u_0 -\nu_i)$}. In the following we will determine the parameter $\lambda_i$. During the computation we will immediately see that $\mathcal{H}_i$ also contains the line from \mbox{ $v^{\sigma}_\pm(u_0)$ to  $v^{\sigma}_\mp(-u_0 +\nu_i) $}. Subsequently we will combine the results in Theorem \ref{Thm:lamda_i}.\\ \\
Let $\langle\cdot,\cdot\rangle_{c_1}$ and $\langle\cdot,\cdot\rangle_{c_3}$  denote the bilinear forms corresponding to the quadrics $\mathcal{C}_1$ and $\mathcal{C}_3$, respectively, i.e., $$
\mathcal{C}_i = \left\{x\in \mathbb{R}^3\colon \left\langle\begin{pmatrix}
x\\1 \end{pmatrix},\begin{pmatrix}
x\\1 \end{pmatrix}\right\rangle_{c_i}=0\right\}.$$
The set $$T_X(\mathcal{H}_i)  =   \left\{x\in \mathbb{R}^3\colon \left\langle\begin{pmatrix}
X\\1 \end{pmatrix},\begin{pmatrix}
x\\1 \end{pmatrix}\right\rangle_{c_1} + \lambda_i \left\langle\begin{pmatrix}
X\\1 \end{pmatrix},\begin{pmatrix}
x\\1 \end{pmatrix}\right\rangle_{c_3}  =0\right\}$$ is
the tangent plane of a quadric $\mathcal{H}_i= \mathcal{C}_1 + \lambda_i\mathcal{C}_3$ at the point $X\in \mathcal{H}_i$.
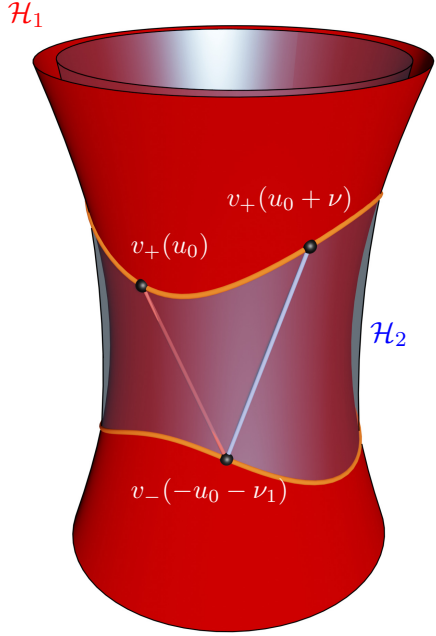
\begin{figure}[t]
\centering
\begin{overpic}[width=.6\linewidth, trim = 0 5cm 0 1cm, clip]
{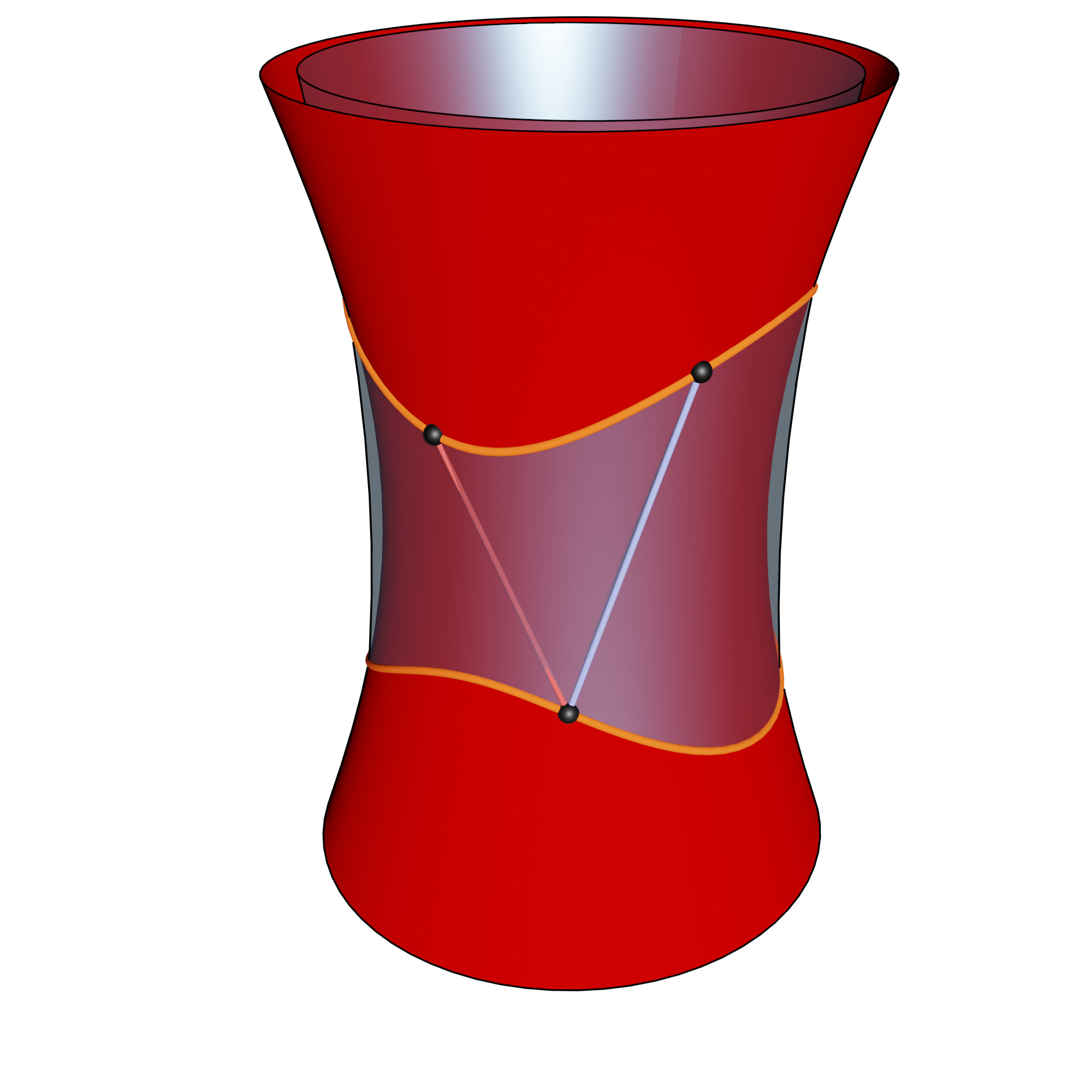} 
\put(38,58){\small \textcolor{white}{$v_+(u_0)$}}
\put(52,65){\small \textcolor{white}{$v_+(u_0+\nu)$}}
\put(38,22){\small \textcolor{white}{$v_-(-u_0-\nu_1)$}}
\put(73,45){\textcolor{blue}{$\mathcal{H}_2$}}
\put(20,92){\textcolor{red}{$\mathcal{H}_1$}}
\end{overpic}
\caption{Case (a): Geometry of the HK-type discrete time Euler top: One iteration of the discrete system maps $v_+(u_0)$ to $v_+(u_0+\nu)$ and can be interpreted as the composition $\tau_{\nu_2}^{a} \circ \tau_{\nu_1}^{a}$ where $\tau_{\nu_i}$ interchanges points on different intersection curve components along generating lines of a hyperboloid $\mathcal{H}_i$.}
\label{Fig:tau}
\end{figure}
\noindent If we solve  \begin{align*} &\left\langle \begin{pmatrix} v^{\sigma}_\pm\left(u_0 \right) \\1 \end{pmatrix}
,\begin{pmatrix} v^{\sigma}_\mp(-u_0 - \nu_i) \\1 \end{pmatrix} \right\rangle_{c_1} + \lambda_i
\left\langle
\begin{pmatrix} v^{\sigma}_\pm\left(u_0 \right)\\1 \end{pmatrix},
\begin{pmatrix} v^{\sigma}_\mp(-u_0 - \nu_i)\\1 \end{pmatrix}  \right\rangle_{c_3} = 0
\end{align*}
for $\lambda_i$, we find the desired ruled quadric in the pencil.\\ \\
We find that for case (a)
\begin{multline}
\label{lamdaeq_a}
\lambda_i(1-F_3)\cn(u_0)\cn(-u_0 - \nu_i) + (1-F_3^{-1}-\lambda_i(F_3-1))\sn(u_0)\sn(-u_0-\nu_i)\\-(1-F_1)\dn(u_0)\dn(-u_0-\nu_i)-(1-F_1)-\lambda_i(1-F_3)=0
\end{multline}
and for case (b)
\begin{multline}
\label{lamdaeq_b}
(1-F_1)\cn(u_0)\cn(-u_0 - \nu_i) +((1-F_1)-\lambda_i F_3(1-F_1))\sn(u_0)\sn(-u_0-\nu_i)\\-\lambda_i(1-F_3)\dn(u_0)\dn(-u_0 -\nu_i)-(1-F_1)-\lambda_i(1-F_3)=0.
\end{multline}
In \cite[Lemma 8.4]{bobenko2020noneuclidean} it has been shown, using addition theorems for Jacobi elliptic functions \cite[§22.8]{Nist}, that, for fixed $v$, an equation of the form $$c_c\cn(u)\cn(u+v) +c_s \sn(u)\sn(u + v) + c_d \dn(u)\dn(u+v)+c_1=0$$ is equivalent to $$c_c\cn(v) +c_d \dn(v) + c_1 = 0$$ for $c_c,c_s,c_d,c_1 \in \mathbb{R}$.
Hence we find that (\ref{lamdaeq_a}) and (\ref{lamdaeq_b}) are equivalent to $$\lambda_i(1-F_3)\cn(\nu_i)-(1-F_1)\dn(\nu_i)-(1-F_1)-\lambda_i(1-F_3)=0$$
and
$$(1-F_1)\cn(\nu_i)-\lambda_i(1-F_3)\dn(\nu_i)-(1-F_1)-\lambda_i(1-F_3)=0,$$
respectively, omitting the sign on $\nu_i$ since $\cn$ and $\dn$ are even functions. This is also the reason why, as previously stated, the quadric has to include the line \mbox{ $v^{\sigma}_\pm(u_0)$ to  $v^{\sigma}_\mp(-u_0 +\nu_i) $}. With $\ns(u) = \frac{1}{\sn(u)}$ we find:
\begin{thm}
\label{Thm:lamda_i} Let $\nu$ be the elliptic time step of the dEt and $\nu_1+\nu_2=\nu$ with $\nu_i\in I$. The solutions of the discrete equations of motion (\ref{equations_of_motion}) are given by the composition \begin{align}\label{composition_tau} \tau^{\sigma}_{\nu_2}\circ\tau^{\sigma}_{\nu_1}\end{align}
with $\tau_{\nu_i}^{\sigma}$ as given in (\ref{tau1}) and (\ref{tau2}). The maps $\tau_{\nu_i}^{\sigma}$ act on the rulings of a quadric $$\mathcal{H}_i := \mathcal{C}_1+\lambda_i\mathcal{C}_3$$
where $\lambda_i$ is given by 
\begin{align*}
\lambda_i &= \frac{1-F_1}{1-F_3} \frac{\dn\left(\nu_i \right)+1 }{\cn\left(\nu_i\right)-1}= - \frac{1-F_1}{1-F_3} \ns^2\left(\frac{\nu_i}{2}\right) \text{ for case (a)} \\
\text{ and }
\lambda_i &= \frac{1-F_1}{1-F_3} \frac{\cn\left(\nu_i \right)-1 }{\dn\left(\nu_i\right)+1} = -\frac{1-F_1}{1-F_3} \sn^2\left(\frac{\nu_i}{2}\right) \text{ for case (b).}
\end{align*}
\end{thm}
\noindent With this theorem we find that the dEt can be described as a composition with one degree of freedom: A fixed choice of $\nu_1$ determines $\nu_2$ uniquely and varying $\nu_1$ (and thus $\nu_2$) leaves the solution of the composition (\ref{composition_tau}) unchanged. The corresponding varying geometry for different values of $\nu_1$ is shown in Figure \ref{Fig:tin01}.
\\ \\ In the following we see how different values of $\nu_i$ determine of what kind the corresponding quadric is and how the signs of the coefficients are given by the case.
First we note that with $\sn^2(u) \in [0,1]$ for $u \in \mathbb{R}$
we have
\begin{align*}
\lambda_i &\in \left[-\frac{1-F_1}{1-F_3},\infty\right] \quad \ \text{ for case (a)} \\
\text{ and } \lambda_i &\in \left[0,-\frac{1-F_1}{1-F_3}\right] \qquad \text{for case (b)}
\end{align*}
and that the quadric $\mathcal{H}_i$ is explicitly given as
$$
\mathcal{H}_i\colon
\lambda_i \delta_2 \delta_3x_1^2  + (1-\lambda_iF_3)\delta_1\delta_3 x_2^2\\ -F_1 \delta_1\delta_2 x_3^2 +(-(1-F_1) - \lambda_i (1-F_3)) = 0.
$$
We will use the notation
\vspace{-.2cm}
\begin{equation}
\label{H_i}\mathcal{H}_i\colon
\mathcal{A}_ix_1^2  + \mathcal{B}_i x_2^2 + \mathcal{C}_i x_3^2 +\mathcal{D}_i= 0.
\end{equation}
Since the quadric $\mathcal{H}_i$ is degenerate if any of its coefficients vanish, the degenerate quadrics are given by $\lambda_i\in \{0,\frac{1}{F_3},-\frac{1-F_1}{1-F_3},\infty\}$:\vspace{.3cm}
\begin{itemize}
\item $\lambda_i = 0 $ corresponds to case (b) with $\nu_i = 0$ and $\mathcal{H}_i = \mathcal{C}_1$,
\item $\lambda_i = \frac{1}{F_3}$ is not possible since $\frac{1}{F_3} < -\frac{1-F_1}{1-F_3}$ for case (a) and $\frac{1}{F_3} > -\frac{1-F_1}{1-F_3}$ for case (b),
\item $\lambda_i =-\frac{1-F_1}{1-F_3}$ corresponds to cases (a) and (b) with $\sn^2\left(\frac{\nu_i}{2} \right) = \ns^2\left(\frac{\nu_i}{2} \right) = 1$, i.e. $\nu_i = 2K_{\sigma}$ where  $\mathcal{H}_i $ degenerates to a cone, 
\item $\lambda_i = \infty $ corresponds to case (a) with $\nu_i = 0$ and $\mathcal{H}_i = \mathcal{C}_3$.
\end{itemize}
\vspace{.3cm}
For $\nu_i \notin \{0,2K_{\sigma}\}$ the quadric $\mathcal{H}_i$ has to be a doubly ruled hyperboloid since with \mbox{$\sn^2\left(\frac{\nu_i}{2} \right)=\sn^2\left(-\frac{\nu_i}{2} \right)$} we have the same $\lambda_i$ for $\nu_i$ and $-\nu_i$ and the two rulings are given by the lines connecting \mbox{$v^{\sigma}_\pm(u_0)$ and $v^{\sigma}_\mp(u_0 + \nu_i) $} and the lines connecting \mbox{ $v^{\sigma}_\pm(u_0)$} and $v^{\sigma}_\mp(u_0 - \nu_i) $.\newpage
\begin{figure}
\begin{minipage}{0.32\linewidth}
\begin{overpic}[width=\linewidth]{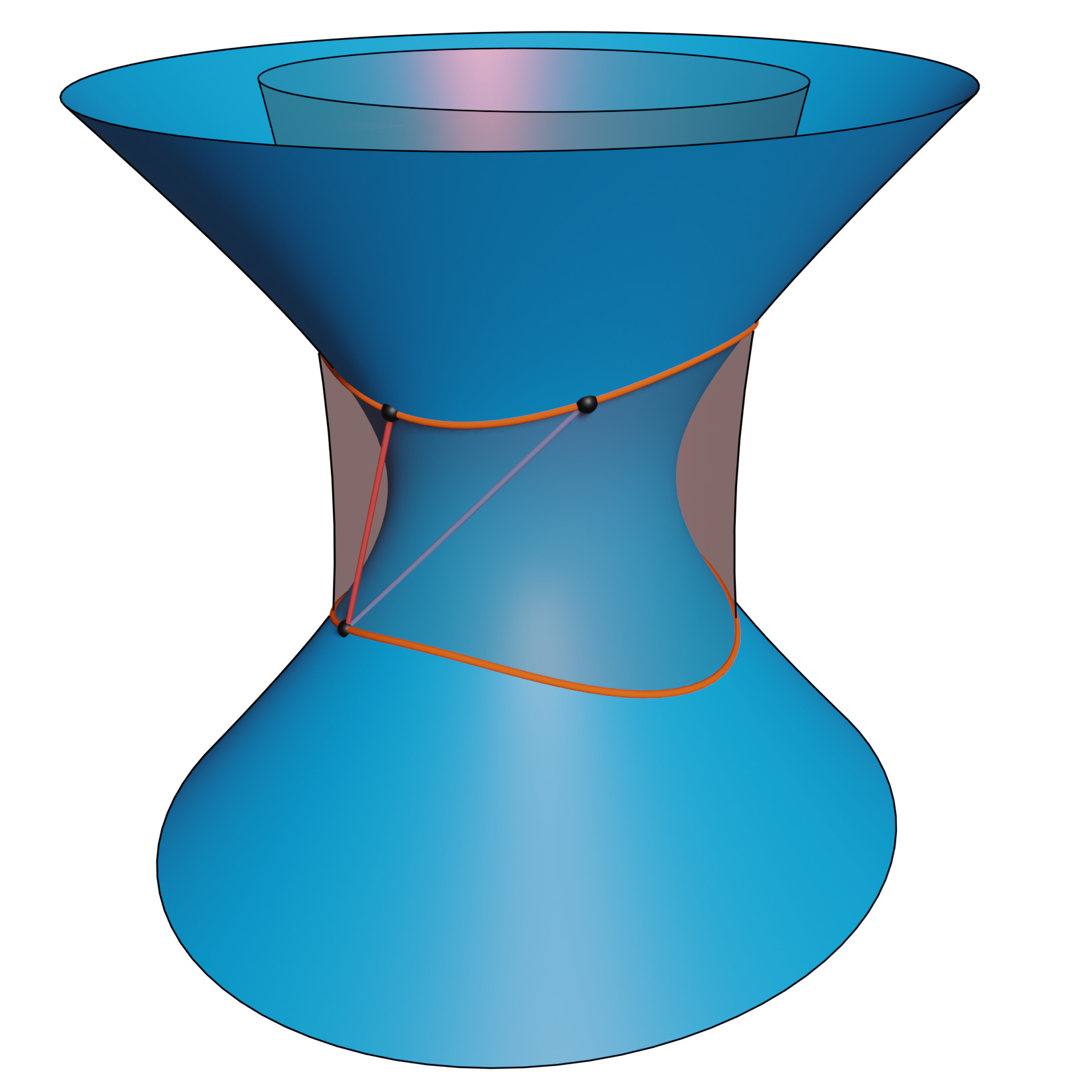}
\put(8,100){\textcolor{cyan}{$\mathcal{H}_2$}}
\put(70,50){\textcolor{red}{$\mathcal{H}_1$}}
\end{overpic}
\end{minipage}
\hfill
\begin{minipage}{0.32\linewidth}
\begin{overpic}[width=\linewidth]{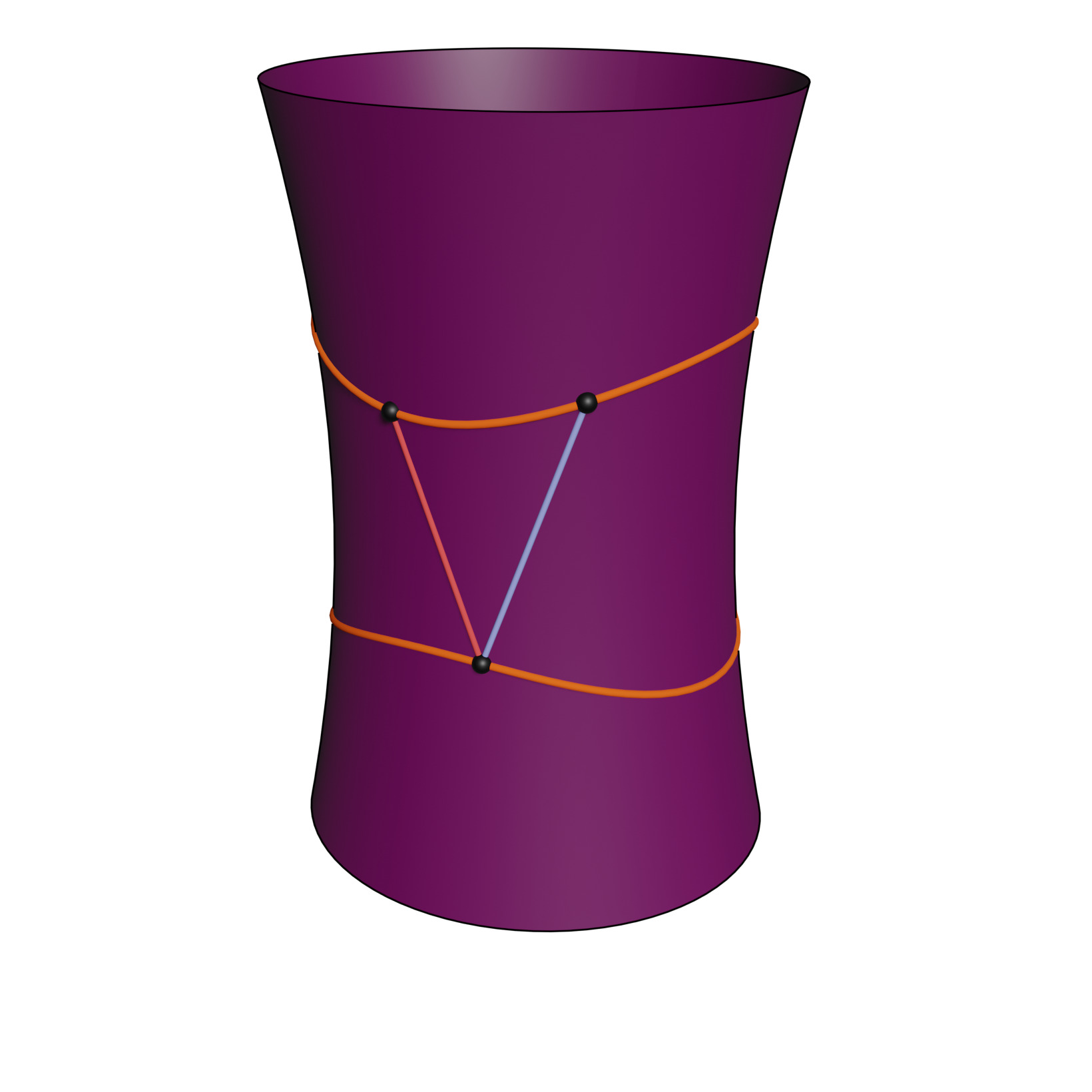}
\put(36,100){\textcolor{magenta}{$\mathcal{H}_1=\mathcal{H}_2$}}
\end{overpic}
\end{minipage}
\hfill
\begin{minipage}{0.32\linewidth}
\begin{overpic}[width=\linewidth]{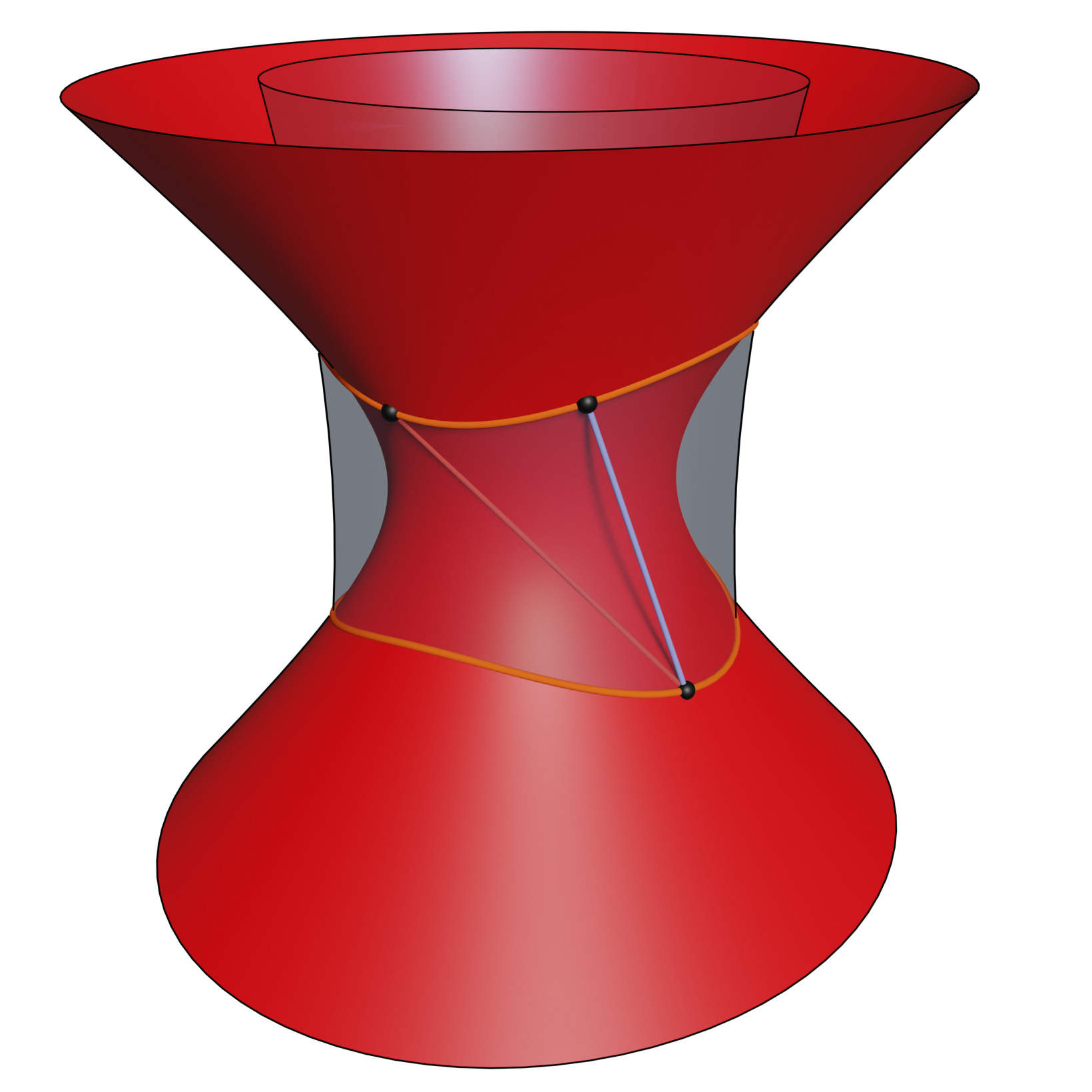}
\put(8,100){\textcolor{red}{$\mathcal{H}_1$}}
\put(70,50){\textcolor{cyan}{$\mathcal{H}_2$}}
\end{overpic}
\end{minipage}
\caption{Case (a): Hyperboloids $\mathcal{H}_1$ (red) and $\mathcal{H}_2$ (blue) for fixed  parameter $\delta \in \mathbb{R}^3$, fixed starting condition $X$ and varying values of $\nu_1$ and $\nu_2$. The pictures correspond to $\nu_1 = -\frac{\nu}{2}$, $\nu_1=\frac{\nu}{2}$ and $\nu_1=\frac{3\nu}{2}$. In the second case it is $\nu_1=\nu_2$ and hence $\mathcal{H}_1$ and $\mathcal{H}_2$ coincide (see Section \ref{Sec:sqrt}).}
\label{Fig:tin01}
\end{figure}
\noindent
For the signs of the coefficients of (\ref{H_i}) we find the following:
\begin{lemma}
\label{lemma:signs}For $\delta_1<0,\delta_2>0,\delta_3<0$ or opposite signs on all $\delta_i$ we find \vspace{-.3cm}
\begin{equation*} 
\begin{aligned}
\mathcal{A}_i &= \lambda_i\delta_2\delta_3 < 0 & \mathcal{B}_i &=(1-\lambda_iF_3)\delta_1\delta_3 \lessgtr 0 \\
\mathcal{C}_i &=-F_1 \delta_1\delta_2>0 & \mathcal{D}_i &= -(1-F_1) - \lambda_i(1-F_3)\gtreqless 0
\end{aligned}
\end{equation*}
\vspace{-.3cm}
where the upper inequality holds for case (a) and the lower for case (b).
\end{lemma}
\begin{proof}
For case (a) we find that
\vspace{-.3cm}
\begin{align*}
F_1 <F_3^{-1} &\Leftrightarrow  1-F_1 >F_3^{-1}(-1+F_3)\\
& \Leftrightarrow   \lambda_i =\frac{1-F_1}{-1+F_3}\ns^2\left(\frac{\nu_i}{2}\right) > F_3^{-1}\\
&\Leftrightarrow  1-\lambda_iF_3<0
\end{align*}
\vspace{-.3cm}
and
\vspace{-.3cm}
\begin{align*}
\lambda_i \geq \frac{1-F_1}{F_3-1}
&\Leftrightarrow -\lambda_i (1-F_3) \geq (1-F_1)\\
&\Leftrightarrow  -(1-F_1) - \lambda_i(1-F_3) \geq 0.
\end{align*}
Reversing the inequalities shows the claim for case (b).
\end{proof}
\noindent Note that for both cases a and (b) we always find \vspace{-.3cm} $$\mathcal{A}_i\mathcal{B}_i\mathcal{C}_i\mathcal{D}_i\geq0$$\vspace{-.3cm}
\vspace{-.3cm}
which allows us to calculate its root and we set \vspace{.4cm}\begin{align}
\label{S_i}
\mathcal{S}_i :=
\pm\sqrt{\mathcal{A}_i\mathcal{B}_i\mathcal{C}_i\mathcal{D}_i}=\delta_1\delta_2\delta_3\sqrt{\lambda_i(1-\lambda_iF_3)(-F_1)(-(1-F_1)-\lambda_i(1-F_3))}
\end{align}
which is positive for $\delta_1<0,\delta_2>0$ and $\delta_3<0$ and negative for all signs of $\delta_i$ reversed.\\
\begin{rem}\textit{Checkerboard incircular nets}\\
The previous construction is closely related to the construction of checkerboard incircular nets. These are quadrilateral nets generated by two families of straight lines in the euclidean plane such that every second quadrilateral admits an incircle, similar to the pattern of a checkerboard. For the construction of a checkerboard incircular net one starts with a cone with its vertex at the origin and the Blaschke cylinder of Laguerre geometry. These two quadrics define a pencil and intersect in a curve consisting of two components. For the construction of a checkerboard incircular net two hyperboloids in this pencil and a starting point on the common intersection curve are chosen. Walking along the generators while alternating between the two hyperboloids in a zig-zag pattern, similar to the iteration in Figure \ref{Fig:tau}, gives rise to one family of points on the intersection curve. This corresponds to one of the families of generating lines of the checkerboard incircular net. The other family can be obtained by starting with a (suitable) point on the other component of the intersection curve. Similar to our cases (a) and (b) one finds that this construction also splits up into two projectively equivalent cases, also distinguished as shown in Figure \ref{Fig:case_a_b}. The first case corresponds to a checkerboard incircular net that is tangent to an ellipse and the latter to one that is tangent to a hyperbola that are both shown in Figure \ref{Fig:CIC}. This and much more details have been investigated in \cite{bobenko2020checkerboard} and extended in \cite{bobenko2020noneuclidean}.\\ \\ \\
\end{rem}
\begin{figure}[h]
\begin{minipage}{0.53\linewidth}
\begin{overpic}[width=.9\linewidth]
{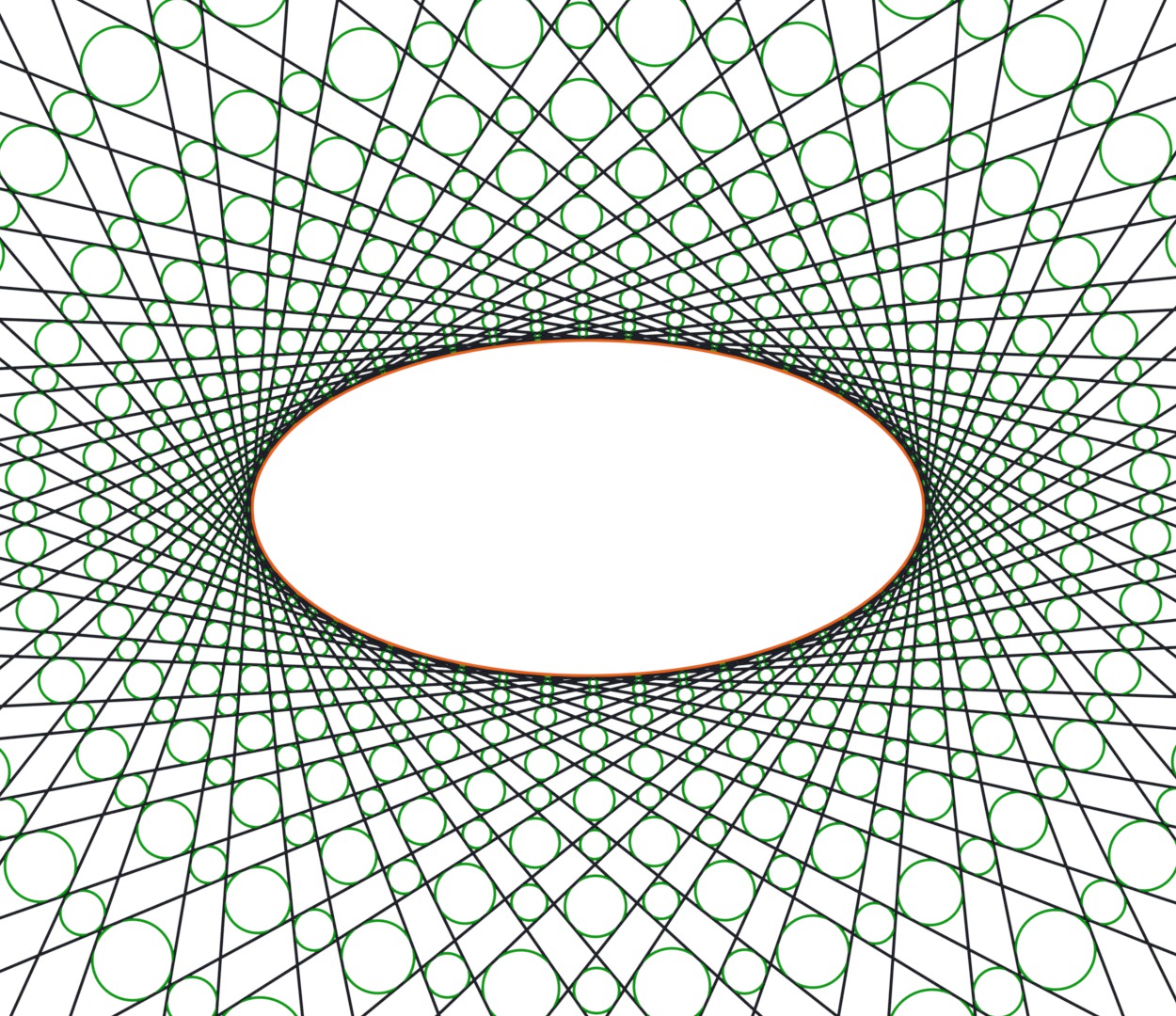}
\end{overpic}
\end{minipage}
\hfill
\begin{minipage}{0.48\linewidth}
\begin{overpic}[width=\linewidth]
{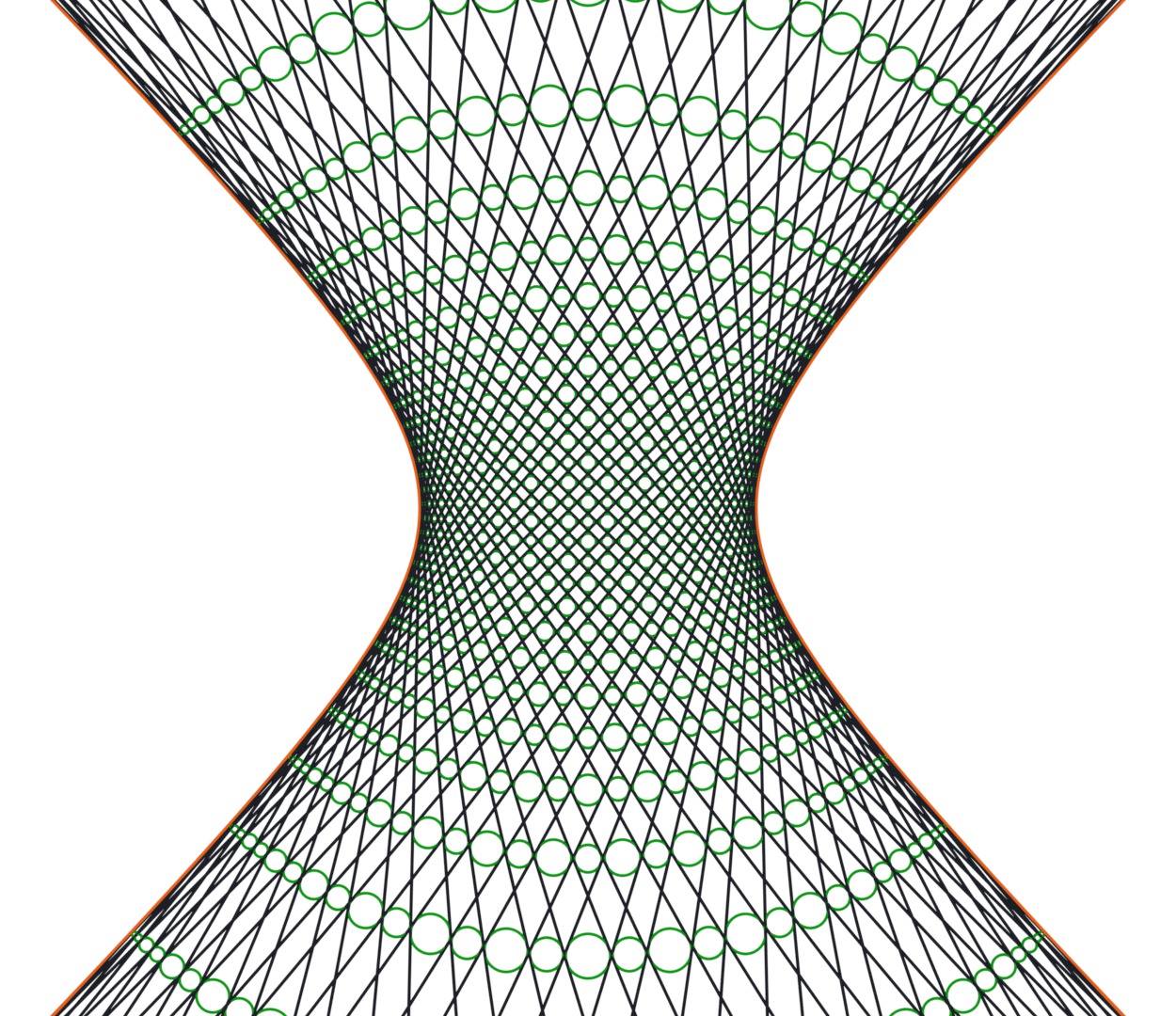} 
\end{overpic}
\end{minipage}
\caption{Checkerboard incircular nets in the Euclidean plane: Two projectively equivalent cases where the generating lines are tangent to an ellipse (left) and tangent to a hyperbola (right), from \cite{bobenko2020checkerboard}.}
\label{Fig:CIC}
\end{figure}
\newpage
\section{The discrete time Euler top as the composition of two complex involutions}
\label{Sec:complex}
Let $$ v^{\sigma}_\pm\colon [0,4K_{\sigma})\rightarrow \mathbb{R}^3$$ denote the parametrization of the intersection curve of the cylinders (\ref{C_i}) for the cases $\sigma \in \{a, b\}$ as given in (\ref{case_a_v}) or (\ref{case_b_v}). In the real case the intersection curve consists of two components.
If we investigate the intersection curve over the complex numbers we find that it consists of a single connected component and constitutes an embedding of an elliptic curve, i.e., a torus, into $\mathbb{C}^3$. It can be parameterized a doubly periodic function with the fundamental parallelogram given by the modulus $k_{\sigma}$ as $$\mathcal{P}_{k_{\sigma}} := \{u+iv \in \mathbb{C}\colon u \in [0,4K_{\sigma})\text{ and } v \in [0,4K'_{\sigma})\}$$ with  $K_{\sigma}$ and $K'_{\sigma}$ being the real and imaginary quarter periods, respectively. In particular $$ K_{\sigma} := \mathcal{K}(k_{\sigma}) \quad \text{and} \quad K'_{\sigma}:=\mathcal{K}(\sqrt{1-k_{\sigma}^2})$$ where $\mathcal{K}(\cdot)$ denotes the elliptic integral of the first kind.
In other words the complex intersection curve is given as the image of the complex function \begin{equation}
\label{def:phi}
\varphi \colon \mathbb{C} \rightarrow \mathbb{C}^3,\quad u+iv \mapsto v^{\sigma}_+(u+iv)
\end{equation}
with the two periods $4K_{\sigma}$ and $4iK'_{\sigma}$. We denote its image as $$\mathfrak{C}:= \varphi(\mathcal{P}_{k_{\sigma}}).$$
The two real components are given by restricting $v$ to either $v\equiv 0$ or $v \equiv 2K'_{\sigma}$, where 
\begin{align*}
&\varphi(u) =  v^{\sigma}_+(u)\\
\text{and} \quad &\varphi(2iK_{\sigma}'-(u + 2K_\sigma)) =  v^{\sigma}_-(-u)
\end{align*}
for $ u \in [0,4K_{\sigma})$, which follows from \begin{align*}
\cn(2iK'-(u+2K)) & = \cn(u) = \cn(-u)\\
\sn(2iK'-(u+2K)) & = \sn(u) = - \sn(-u)\\
\dn(2iK'-(u+2K)) & = -\dn(u) = -\dn(-u).
\end{align*} 
The fundamental domain and the preimages of $v^\sigma_+$ and $v^\sigma_-$ are shown in Figure \ref{Fig:fund_domain}.

\begin{figure}
\centering
\begin{overpic}[width=.7\linewidth]
{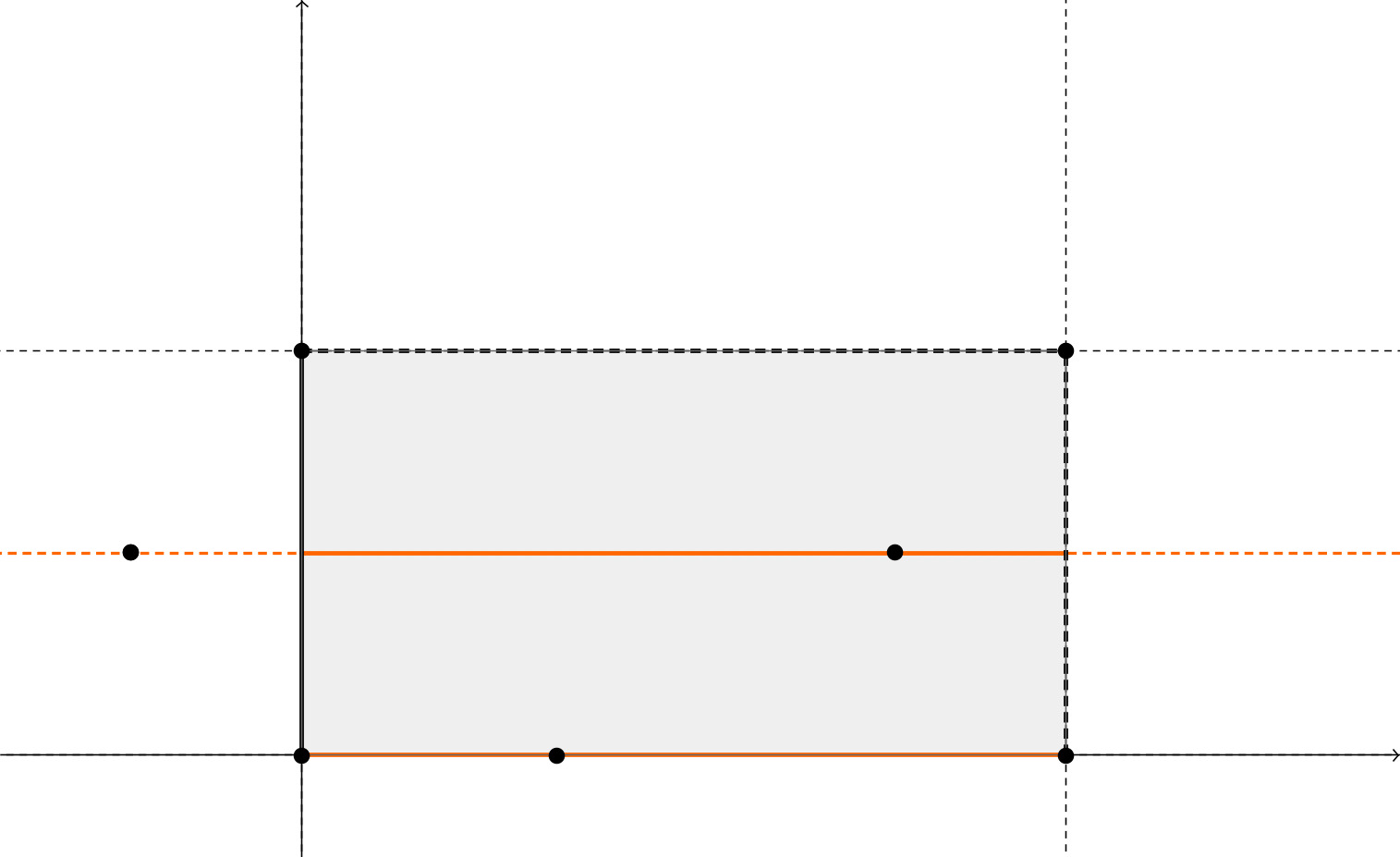}
\put(39,4){$u $}
\put(2,23){$\scriptstyle 2iK_\sigma'-u - 2K_\sigma$}
\put(57,23){$\scriptstyle 2iK_\sigma'-u + 2K_\sigma$}
\put(12,37){$4iK_\sigma'$}
\put(77,3){$4K_\sigma$}
\put(85,50){$\mathbb{R}^2\cong \mathbb{C}$}
\end{overpic}
\caption{The fundamental domain $\mathcal{P}_{k_\sigma}$ of the complex function $\varphi$. The two orange segments in the fundamental domain represent the preimages of the real components $v_+$ and $v_-$. Due to periodicity of $\varphi$ the images of $2iK_\sigma'-u -2K_\sigma$ and $2iK_\sigma'-u + 2K_\sigma$ coincide.}
\label{Fig:fund_domain}
\end{figure}
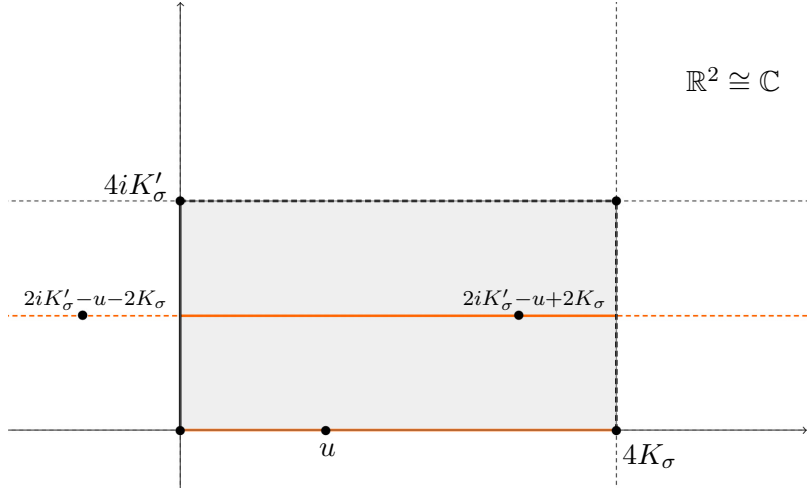
\begin{prop}
For given $\nu_i\in \mathbb{R}$ the maps \begin{align*}
\iota^{\pm}_{\nu_i}&\colon \mathfrak{C} \rightarrow \mathfrak{C},\quad \varphi(z) \mapsto \varphi(2iK'_{\sigma}-(z+2K_\sigma) \pm \nu_i)
\end{align*} on the image of $\varphi$, given in (\ref{def:phi}), are involutions.
\end{prop}
\begin{proof}
It holds
\begin{align*}(\iota^\pm_{\nu_i}\circ \iota^\pm_{\nu_i} )(\varphi(z)) &= \iota^\pm_{\nu_i}(\varphi(2iK'_{\sigma}-(z+2K_\sigma) \pm \nu_i))\\
 &= \varphi(2iK'_{\sigma}-(2iK'_{\sigma}-(z+2K_\sigma) \pm \nu_i +2K_\sigma) \pm \nu_i) \\
 & =\varphi(z) .
\end{align*}
\end{proof} 
\noindent Note that if $z = u$ or $z =2iK'_{\sigma}-(u+2K_\sigma)$ for $u \in \mathbb{R}$ these involutions describe walking from one of the real components of the intersection curve to the other. The lines connecting real preimages and images under these involutions lie on the hyperboloid $\mathcal{H}_i$, that is determined by $\nu_i$. One should be aware that the choice of $\pm$ in $\iota^{\pm}_{\nu_i}$ is not related to the different real components as in the maps $v^{\sigma}_\pm$ but rather determines along which of the generating lines of $\mathcal{H}_i$ we walk. In particular they describe the maps $\tau_{\nu_i}$ from (\ref{tau1}) and (\ref{tau2}) as involutions, see Figure \ref{Fig:complex_preim_and_im}.\newpage
\begin{thm} For given parameter $\delta$ of the dEt let $\nu$ be the corresponding elliptic time step and $\nu_1+\nu_2=\nu$ with $\nu_i\in I$. Then the composition $\iota_{\nu_2}^+\circ \iota_{\nu_1}^- $ of the involutions 
\begin{align*}
\iota_{\nu_1}^-&\colon Re(\mathfrak{C}) \rightarrow Re(\mathfrak{C}),\quad \varphi(z) \mapsto \varphi(2iK'_{\sigma}-(z+2K_\sigma) - \nu_1)\\
\text{ and }\iota_{\nu_2}^+&\colon Re(\mathfrak{C}) \rightarrow Re(\mathfrak{C}),\quad \varphi(z) \mapsto \varphi(2iK'_{\sigma}-(z+2K_\sigma) + \nu_2)
\end{align*}
on the real part of the image of $\varphi$, given in (\ref{def:phi}), coincides with the birational map $f(\cdot,\delta)$ from (\ref{explicit_f}).
\end{thm}
\begin{proof}
The statement is already verified by Theorem \ref{Thm:lamda_i}
but also can easily be shown by direct computation. It holds
\begin{align*}
(\iota_{\nu_2}^+\circ \iota_{\nu_1}^- )(\varphi(z)) &= \iota_{\nu_2}^+(\varphi(2iK'_{\sigma}-(z+2K_\sigma) - \nu_1))\\
 &= \varphi(2iK'_{\sigma}-(2iK'_{\sigma}-(z+2K_\sigma) - \nu_1 +2K_\sigma) + \nu_2) \\
 & = \varphi(z+\nu)
\end{align*}
and for $\varphi(z) \in Re(\mathfrak{C})$ we have $\varphi(z) =  v^{\sigma}_\pm(\pm u)$and therefore $$(\iota_{\nu_2}^+\circ \iota_{\nu_1}^- )(v^{\sigma}_\pm(\pm u)) =  v^{\sigma}_\pm(\pm u + \nu) $$ where the choice of $\pm$ given by $z = u$ or $z = 2iK'_{\sigma}-(u+2K_\sigma)$, respectively.
\end{proof}
\begin{figure}
\begin{minipage}{.49\linewidth}
\begin{overpic}[width=\linewidth]
{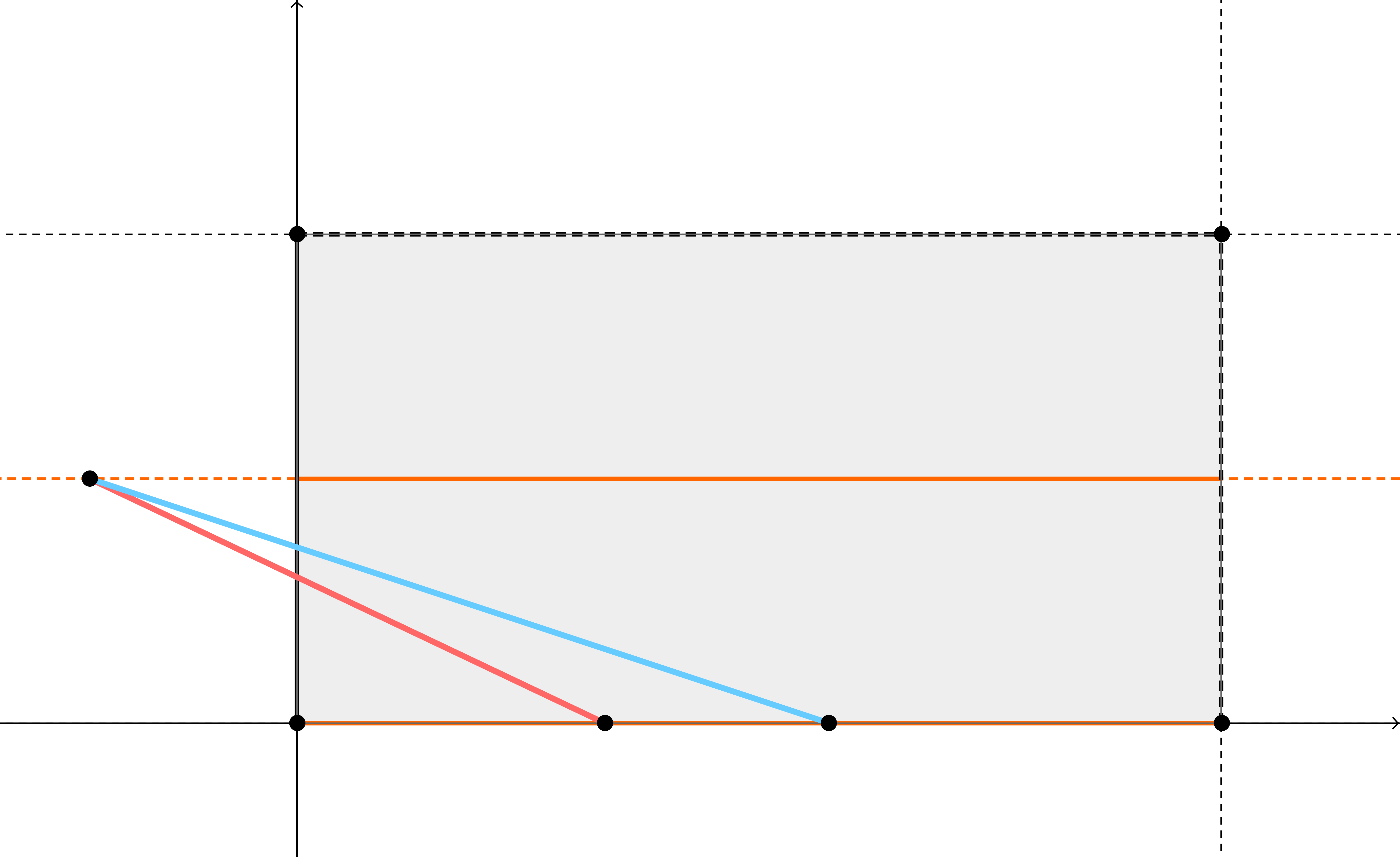}
\put(43,5){$u$}
\put(55,5){$ u+\nu $}
\put(0,28){$\scriptstyle 2iK'-(u + 2K) -\nu_1$}
\put(10,47){$4iK'$}
\put(87,4){$4K$}
\end{overpic}
\end{minipage}
\hfill
\begin{minipage}{.49\linewidth}
\begin{overpic}[width=\linewidth]
{walking_generators} 
\put(38,63){\small \textcolor{white}{$\varphi(u)$}}
\put(52,70){\small \textcolor{white}{$\varphi(u+\nu)$}}
\put(33,29){\small \textcolor{white}{$\scriptstyle \varphi(2iK'-(u + 2K) -\nu_1)$}}
\end{overpic}
\end{minipage}
\caption{Preimages and images of the composition of the involutions $\iota^-_{\nu_1}$ (red lines) and $\iota^+_{\nu_2}$ (blue lines).}
\label{Fig:complex_preim_and_im}
\end{figure}
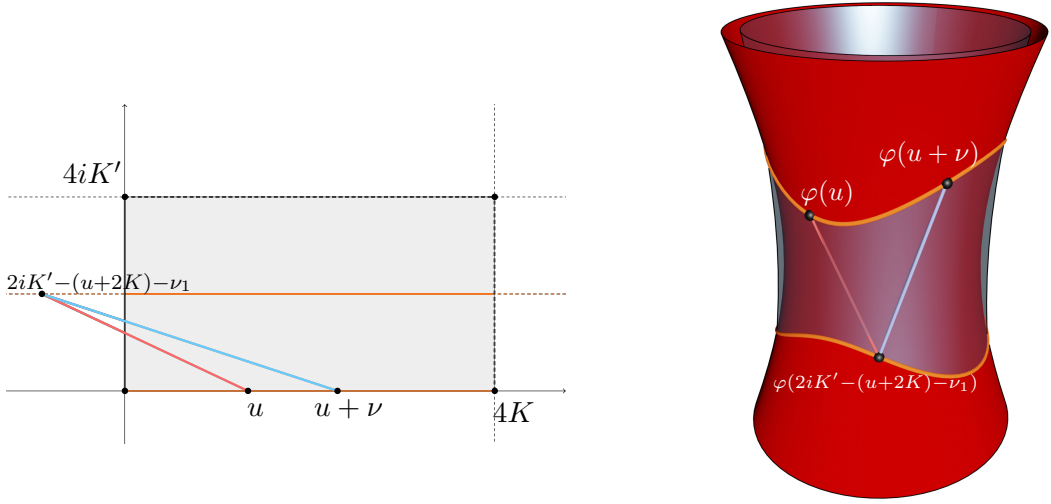
\begin{rem}
As an elliptic space curve it is possible to define a group (addition) law on the image $\mathfrak{C}$ of $\varphi$. The underlying geometric constructions of these group law have been presented in \cite[p. 20-22]{husemoller1987elliptic}. This addition law corresponds to the remarkable property of Jacobi elliptic functions $$\begin{vmatrix}
\cn(z_1) & \sn(z_1) & \dn(z_1) & 1\\
\cn(z_2) & \sn(z_2) & \dn(z_2) & 1\\
\cn(z_3) & \sn(z_3) & \dn(z_3) & 1\\
\cn(z_4) & \sn(z_4) & \dn(z_4) & 1\\
\end{vmatrix} = 0 \  \text{ for } \  z_1+z_2+z_3+z_4 = 0$$
for a common modulus $k$.
From this we also find that the four real points\\  \mbox{$v_+(u_0),v_-(u_0+ \nu_i),v_-(\tilde{u}_0)$ and $v_+(\tilde{u}_0 +\nu_i)$} are coplanar, see also \cite[Proposition 8.4.]{bobenko2020noneuclidean}.
\end{rem}
\vspace{.4cm}

\section{The discrete time Euler top as the composition of real and rational involutions}
\label{Sec:real_inv}
In Section \ref{Sec:geometric_maps} we have shown that the discrete time Euler top can be described as the composition of two maps acting on two ruled quadrics that are members of the pencil that is introduced by the conserved quantities. These maps, given in (\ref{tau1}) and (\ref{tau2}), are formulated in terms of elliptic functions and they are no involutions. In the previous section we found a representation of these maps as involutions, formulated with the help of a complex valued, doubly periodic function. The goal of this section is to derive a representation by real and rational functions that are also involutions. Therefore, in the first subsection we will introduce the geometric and analytic construction of such involutions defined on the intersection curve of general hyperboloids and cylinders. Geometrically, they act on the rulings of two non-degenerate, doubly ruled hyperboloids. By specializing them to the quadrics accociated to the dEt, which we described in Section \ref{Sec:geometric_maps}, we find compositions of two real involutions describing the iterations of the dEt. Finally, we treat various examples, firstly the square root of the dEt and secondly birational examples of the involutions. The latter correspond to the cases where one of the ruled quadrics is degenerate and since they are explicitly treated in Section \ref{Sec:birational} we omit those cases in the following general construction.
\vspace{.4cm}´
\subsection{Involutions on the intersection of hyperboloids and cylinders}
\vspace{.4cm}
\label{Sec:geo_background}

We start considering general one-sheeted hyperboloids. As a doubly ruled surface we can parameterize a one-sheeted hyperboloid by \begin{align}
\label{H_para} s :\mathbb{R}^2 \rightarrow \mathbb{R}^3, \quad(u,v)\mapsto b(u) + vd(u) 
\end{align}
where $b$ denotes a directrix on the hyperboloid and $d$ a direction vector. There exist two families of direction vectors $d_\pm$, each corresponding to one ruling of the hyperboloid. A fixed choice of one family gives rise to a parametrization as above. Through a fixed point on a one-sheeted hyperboloid there exist two distinct lines, belonging to distinct rulings, that are fully contained in the hyperboloid. We will denote these lines as the \emph{generating lines} though the given point.
We will start by giving an explicit description of $d_\pm$, the direction vectors that span these lines:
\begin{prop}
\label{Thm:d(x)}
For a point $X = (
X_1,X_2,X_3)\in \mathbb{R}^3$ on a one-sheeted hyperboloid $\mathcal{H}\subset \mathbb{R}^3$ given by 
\begin{align}
\label{H}
\mathcal{H}\colon A x_1^2 + B x_2^2 +C x_3^2+D = 0 \quad \text{ with }\quad  A,B,C,D \in \mathbb{R}
\end{align} 
generating lines through $X$ are given by $$
v \mapsto X + vd_\pm(X)$$ \vspace{-.3cm}with\vspace{-.3cm}
\begin{align*}
d_\pm (X):=
&\begin{pmatrix}
a_\mp(X)\\
b_\pm(X)\\
c(X)
 \end{pmatrix}=
\begin{pmatrix}
ABC X_1X_3 \mp \sqrt{ABCD} B X_2\\
ABC X_2X_3 \pm \sqrt{ABCD} AX_1\\
-AB(AX_1^2+BX_2^2)
\end{pmatrix}\\
\end{align*}
\end{prop}
\begin{rem}
Alternative forms for $d_\pm (X)$ are 
\begin{align*}
d^{(1)}_\pm (X) &= 
\begin{pmatrix}
-BC(BX_2^2+CX_3^2)\\
ABC X_1X_2 \mp \sqrt{ABCD} C X_3\\
ABC X_1X_3 \pm \sqrt{ABCD} B X_2
\end{pmatrix} \\ \text{or} \qquad  d^{(2)}_\pm (X)&= 
\begin{pmatrix}
ABC X_1X_2 \pm  \sqrt{ABCD} C X_3\\
-AC(AX_1^2+CX_3^2)\\
ABCX_2X_3 \mp \sqrt{ABCD} A X_1
\end{pmatrix} .
\end{align*}
\end{rem}
\begin{proof}
Let $\langle\cdot,\cdot\rangle_\mathcal{H}$ denote the quadric form given by $\mathcal{H}$, i.e., $$
\mathcal{H} = \left\{x\in \mathbb{R}^3\colon \left\langle\begin{pmatrix}
x\\1 \end{pmatrix},\begin{pmatrix}
x\\1 \end{pmatrix}\right\rangle_{\mathcal{H}}=0\right\}$$ and $\parallel x \parallel_\mathcal{H}$ the corresponding norm. If we omit the $\pm$ in the notation, the direction vector $d(X)$ has to fulfill $$\left\lVert \begin{pmatrix}
X\\1 \end{pmatrix} + v \begin{pmatrix}
d(X)\\0 \end{pmatrix}\right\rVert_\mathcal{H}=0 \quad \text{for all } v \in \mathbb{R}.$$
This consideration yields a set of two equations:
 \begin{equation}
 \label{Eq:dx}
\begin{aligned}
A a^2(X) + B b^2(X) + C c^2(X) &=0\\
A a(X)X_1 + Bb(X)X_2 + C c(X) X_3 &= 0.
\end{aligned}
\end{equation}
They leave one degree of freedom since multiplication of $a(X),b(X)$ and $c(X)$ with a scalar gives rise to the same equations.
Solving for $a(X),b(X)$ and $c(X)$ leads to $d(X)=d_\pm(X)$ as given in the theorem or any $d^{(i)}_\pm(X)$ from the remark.
\end{proof}
\noindent Let $\mathcal{C}$ be a cylinder given by \begin{align}
\label{C}
\mathcal{C} \colon \alpha x_1^2 +\beta x_3^2 + \gamma = 0\text{ with }\quad  \alpha, \beta,\gamma \in \mathbb{R} \end{align}
and $\mathcal{H}$ be the one-sheeted hyperboloid as given in (\ref{H}).
The intersection curve $\mathcal{H}\cap\mathcal{C}$ of these two quadrics in real space consists of two connected components. Using one of the components as a directrix $b(u)$ and one family of direction vectors $d_\pm(b(u))$ from Proposition \ref{Thm:d(x)}, we can derive a parametrization of the hyperboloid as described in (\ref{H_para}). 
Each of the two generating lines through a point lying on the directrix intersects the cylinder in exactly one other point, lying on the other component of the intersection curve, see Figure \ref{Fig:iota}. We will now explicitly describe the involutions $\iota_{\mathcal{H},\mathcal{C}}^\pm$ interchanging these two points for a fixed cylinder $\mathcal{C}$ and hyperboloid $\mathcal{H}$.
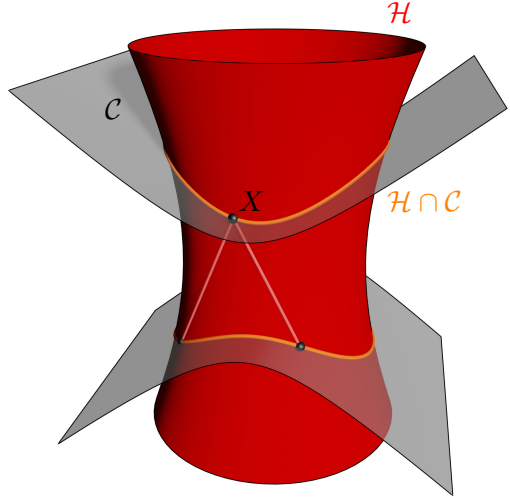
\begin{figure}
\centering
\begin{overpic}[width=.45\linewidth]
{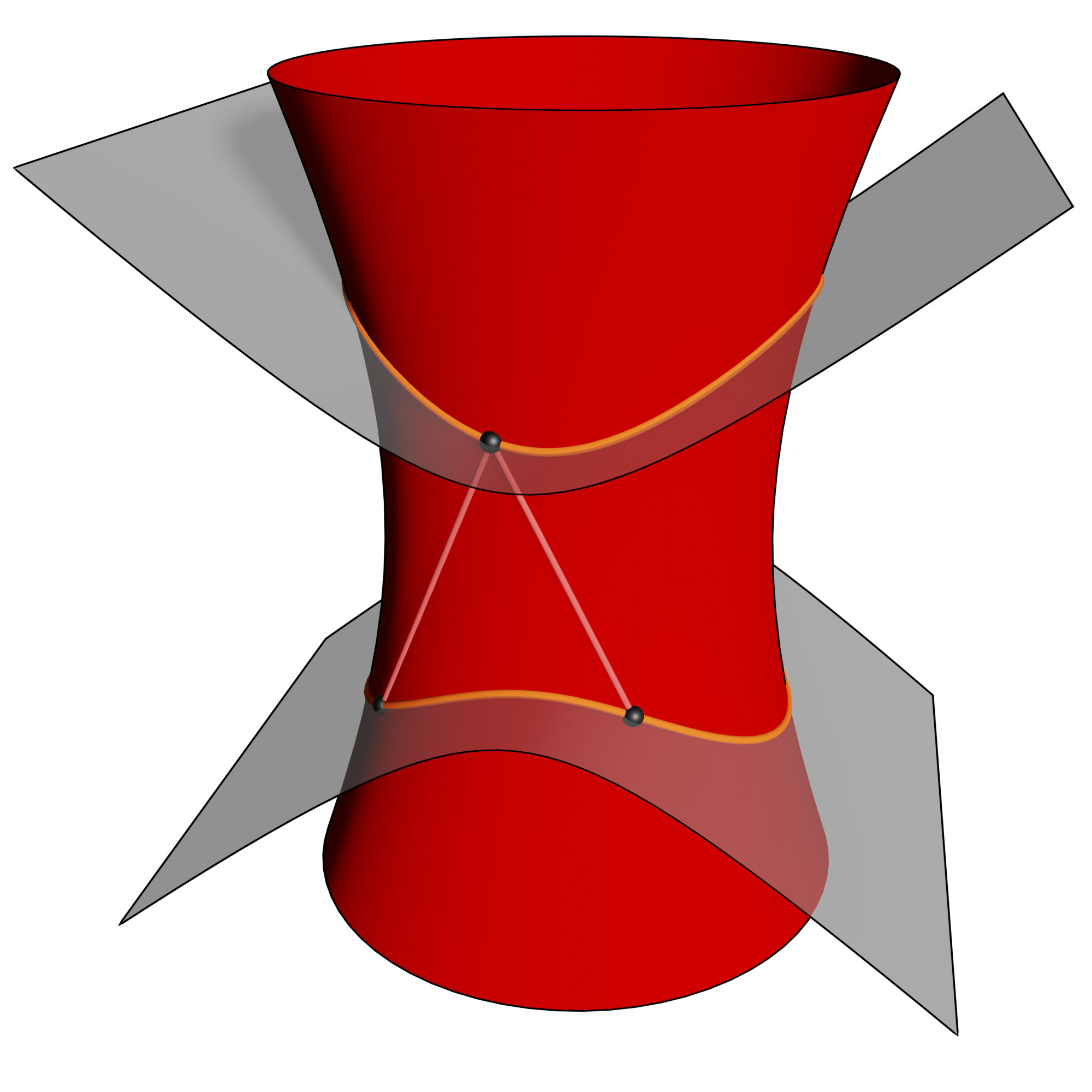} 
\put(46,61){$X$}
\put(75,98){\textcolor{red}{$\mathcal{H}$}}
\put(75,61){\textcolor{orange}{$\mathcal{H}\cap \mathcal{C}$}}
\put(20,79){$\mathcal{C}$}
\end{overpic}
\caption{The intersection curve (orange) of a one-sheeted hyperboloid $\mathcal{H}$ (red) and a hyperbolic cylinder $\mathcal{C}$ (gray) and intersection points of the generating lines through $X$ with $\mathcal{C}$}
\label{Fig:iota}
\end{figure}
\begin{prop}
\label{Thm:NdIntersectionPoint}
The generating lines through a point $X = (X_1,X_2,X_3)$, lying on the intersection curve $\mathcal{H} \cap \mathcal{C}$ of a one-sheeted hyperboloid $\mathcal{H}$ and a cylinder $\mathcal{C}$ given by (\ref{H}) and (\ref{C}), intersect the cylinder $\mathcal{C}$ in a second point other than $X$, namely $$ \hat{X}= X + v(X) d_\pm(X)$$
where
$$v(X) := -2 \frac{\alpha a_\mp(X)X_1 +\beta c(X) X_3}{\alpha a_\mp^2(X) +\beta c^2(X)} $$
and $d_\pm(X)$ as given in Proposition \ref{Thm:d(x)}.
\end{prop}

\begin{proof}
Since the generating lines through $X$ are given by $X+v d_\pm(X)$ for $v \in \mathbb{R}$, it suffices to substitute $X+vd_\pm(X)$ into $\mathcal{C}$ and solve for $v$, while using that $X$ already lies on $\mathcal{C}$.
\end{proof}
\begin{prop}
The maps \begin{align}
\iota_{\mathcal{H},\mathcal{C}}^\pm &\colon \mathcal{H}\cap \mathcal{C}\rightarrow \mathcal{H}\cap \mathcal{C},\quad x\mapsto x + v(x)d_\pm(x)
\label{Inv:iota}
\end{align} with $\mathcal{H},\mathcal{C},v(x)$ and $d_\pm(x)$ as given in Proposition \ref{Thm:NdIntersectionPoint} are involutions, i.e., $$(\iota_{\mathcal{H},\mathcal{C}}^\pm \circ \iota_{\mathcal{H},\mathcal{C}}^\pm) (x) = x.$$
\end{prop}
\begin{proof}
Let $x \in \mathcal{H}\cap \mathcal{C}$ and $\hat{x}:= \iota^\pm_{\mathcal{H},\mathcal{C}}(x) =x + v(x)d_\pm(x)$ for a fixed choice of $\pm$. We want to show that the second iteration of $\iota^\pm_{\mathcal{H},\mathcal{C}}$ takes $\hat{x}$ back to $x$. This is the case iff $d_\pm(x)$ and $d_\pm(\hat{x})$ for the same choice of $\pm$ are linearly dependent. To prove their dependency we show that their cross product equals to zero. For simplicity we introduce the notation $v:= v(x)$ and  \begin{align*}
\begin{pmatrix}
a\\
b\\
c
 \end{pmatrix}:&=d_\pm (x)
=
\begin{pmatrix}
ABC x_1x_3 \mp \sqrt{ABCD} B x_2\\
ABC x_2x_3 \pm \sqrt{ABCD} Ax_1\\
-AB(Ax_1^2+Bx_2^2)
\end{pmatrix}\end{align*}
Then \vspace{-.2cm} $$
d_\pm (\hat{x})=d_\pm(x+v(x)d_\pm(x))=
\begin{pmatrix}
ABC (x_1+va)(x_3+vc) \mp \sqrt{ABCD} B (x_2+vb)\\
ABC (x_2+vb)(x_3+vc)\pm \sqrt{ABCD} A(x_1+va)\\
-AB(A(x_1+va)^2+B(x_2+vb)^2)
\end{pmatrix}$$
and \begin{align*}
(d_\pm(x)\times d_\pm(\hat{x}))_1 = &-AB b (A(x_1+va)^2+B(x_2+vb)^2)\\
&-ABC c(x_2+vb)(x_3+vc)\mp \sqrt{ABCD} A(x_1+va)\\
= &-AB b (Ax_1^2+Bx_2^2)-ABCcx_2x_3\mp \sqrt{ABCD}c x_1\\
&+ v(-2ABb(Aax_1+Bbx_2)-ABCc(cx_2+bx_3)\mp \sqrt{ABCD} A ac) \\
&+v^2(-ABb(Aa^2+Bb^2+Cc^2))\\
\stackrel{(\ref{Eq:dx})}{=}  &\ bc - cb + v c(ABCbx_3-ABCcx_2\mp  \sqrt{ABCD} A a)\\
= &\ v c A^2B^2 C x_2 (Ax_1^2 +Bx_2^2+Cx_3^2+D)
= 0 
\end{align*}
By similar calculation for the remaining two entries we can show that $$ d_\pm(x)\times d_\pm(\hat{x}) = \begin{pmatrix}
(d_\pm(x)\times d_\pm(\hat{x}))_1\\
(d_\pm(x)\times d_\pm(\hat{x}))_2\\
(d_\pm(x)\times d_\pm(\hat{x}))_3
\end{pmatrix} = 
\begin{pmatrix}
0\\0\\0
\end{pmatrix}$$
which proves the claim.
\end{proof}
\vspace{.4cm}

\subsection{The discrete time Euler top as the composition of real involutions}
\vspace{.4cm}
\label{Sec:real_inv}

As motivated in the beginning of Section \ref{Sec:real_inv} this section will give a general formulation of a pair of involutions representing the dEt. The desired involutions will be stated in terms of real functions omitting elliptic functions entirely. This formulation will hold for all generic, non-degenerate cases and the degenerate cases are explicitly investigated in Section \ref{Sec:birational}. \\ \\  
Let $\delta\in \mathbb{R}^3$ be given as the parameter of the dEt with the signs of its entries as introduced in (\ref{sign_alphs}). For a given initial condition $X\in \mathbb{R}^3$ Proposition \ref{Prop:nu} defines the elliptic time step $\nu>0$. Let $\nu_1$ and $\nu_2$ such that $$\nu_1+\nu_2=\nu \quad \text{ with} \quad \nu_i \in \tilde{I} := [-2K_{\sigma},2K_{\sigma}]\setminus \{0,\pm2K_{\sigma}\}.$$ 
Excluding the values $\{0,\pm2K_{\sigma}\}$ ensures that the ruled quadrics defined by $\nu_i$ are non-degenerate, i.e., one-sheeted hyperboloids.
\\
\\
In Section \ref{Sec:geo_background}
we defined the involutions \begin{align*}&\iota_{\mathcal{H},\mathcal{C}}^\pm \colon \mathcal{H}\cap \mathcal{C}\rightarrow \mathcal{H}\cap \mathcal{C},\quad x\mapsto x + v(x)d_\pm(x) \\ \\
\text { with }d_\pm (x)&:=
\begin{pmatrix}
a_\mp(x)\\
b_\pm(x)\\
c(x)
 \end{pmatrix}=
\begin{pmatrix}
ABC x_1x_3 \mp \sqrt{ABCD} Bx_2\\
ABC x_2x_3 \pm \sqrt{ABCD} Ax_1\\
-AB(Ax_1^2+Bx_2^2)
\end{pmatrix} \\
\text{and }
v(x) &:= -2 \frac{\alpha a_\mp(x)x_1 +\beta c(x) x_3}{\alpha a_\mp^2(x) +\beta c^2(x)}  \end{align*}
on the intersection curve $\mathcal{H}\cap\mathcal{C}$ of a hyperboloid $\mathcal{H}$ and a cylinder $\mathcal{C}$. Explicitly reformulating these maps for the quadrics introduced by the geometry of the dEt, i.e., $\mathcal{C}_2$ from ($\ref{C_i}$) and $\mathcal{H}_i$, given by ($\ref{H_i}$), leads to the involutions
 \begin{align*}
 &\iota_{\mathcal{H}_i,\mathcal{C}_2}^{\pm \delta} \colon \mathcal{H}_i\cap \mathcal{C}_2\rightarrow \mathcal{H}_i\cap \mathcal{C}_2,\quad x\mapsto x + v(x)d(x) \\ \\
\text { with }
 d(x) &:= \begin{pmatrix}
a(x)\\
b(x)\\
c(x)
\end{pmatrix}=
\begin{pmatrix}
\mathcal{A}_i
\mathcal{B}_i
\mathcal{C}_i x_1x_3 -
\mathcal{S}_i
\mathcal{B}_i x_2\\
\mathcal{A}_i
\mathcal{B}_i
\mathcal{C}_i x_2x_3 +
\mathcal{S}_i
\mathcal{A}_ix_1\\
-\mathcal{A}_i
\mathcal{B}_i(
\mathcal{A}_ix_1^2+
\mathcal{B}_ix_2^2)
\end{pmatrix}\\ \text{and }
v(x) &:= -2 \frac{F_2\delta_2\delta_3a(x)x_1 -\delta_1\delta_2 c(x) x_3}{F_2\delta_2\delta_3 a(x)^2 -\delta_1\delta_2 c(x)^2} \end{align*}
where $\mathcal{S}_i$, as defined in (\ref{S_i}), changes its sign by reversing all signs of $\delta_i$.
As $\mathcal{H}_i$ and $\mathcal{C}_2$ are derived from the initial parameters of the dEt, we can set \begin{align*}
\iota_{(\nu_i,\pm\delta)}:=\iota^{\pm\delta}_{\mathcal{H}_i,\mathcal{C}_2}
\end{align*} to reinforce that these involutions only depend on $\nu_i$ and $\pm\delta$.
\noindent
\begin{thm}
\label{Thm:real_involution}
For given parameter $\delta$ of the dEt let $\nu$ be the corresponding elliptic time step and $\nu_1+\nu_2=\nu$ with $\nu_i\in \tilde{I}$. Then the compositions of involutions 
\begin{align}\begin{cases}
\label{iotainversion}
\iota_{(\nu_2,-\sgn(\nu_2) \delta)}\circ \iota_{(\nu_1,\sgn(\nu_1)\delta)} & \text{ for case (a)}\\
\iota_{(\nu_2,\sgn(\nu_2) \delta)}\circ \iota_{(\nu_1,-\sgn(\nu_1)\delta)} & \text{ for case (b)}\end{cases}
\end{align} coincide with the birational map $f(\cdot,\delta)$ from (\ref{explicit_f}).
Inverting both signs leads to the inverse of $f$.

\end{thm}
\noindent For values of $\nu_i$ where the quadric (\ref{H_i}) degenerates, namely $\nu_i=0$ or $\nu_i=2K_{\sigma}$, the corresponding involutions are birational since $\mathcal{S}_i$ vanishes. The (trivial) maps for these cases are explicitly given in Section \ref{Sec:birational}. For all other values of $\nu_i$ the choice of the sign of the parameter $\delta$ of the involutions $\iota_{(\nu_i,\pm\delta)}$ is visualized in Figure \ref{Fig:sgn}.
\begin{proof}
By the geometric construction it is clear that two such involutions exist and give rise to $f$. It only remains to show that the compositions induce a walking along the correct family of generators of the hyperboloids. This is equivalent to the correct choice of the sings of the input parameters $\delta$. We give a proof why the pair of signs $(\sgn(\nu_1), -\sgn(\nu_2))$ will fulfill the claim for case (a). Case (b) follows in a similar way. \\ \\Step 1: $\iota_{(\nu_1,\sgn(\nu_1)\delta)} = \tau_{\nu_1}^{a}$\\
We know, that  the given involution $\iota_{(\nu_1,\pm \delta)}$ corresponds to either  $\tau_{\nu_1}^{a}$ or $\tau_{-\nu_1}^{a}$ from (\ref{tau1}).
We distinguish the cases $\nu_1 \gtrless 0$. Within thesese cases a continous deformation of the parameter $\nu_1$ will lead to a continous deformation of $\iota$ and $\tau$. Hence it suffices to show the claim for a fixed $\nu_1>0$ and a fixed $\nu_1<0$ which proves it, by continuity, for all $\nu_i\in \tilde{I}$.\\
Step 1.1. : $\nu_1>0$\\
We choose the value $\nu_1=\nu$ where we know that $$\tau^{a}_{\nu_1}(x)=\begin{pmatrix}
f(x,\delta)_1\\
f(x,\delta)_2\\
-f(x,\delta)_3\end{pmatrix}.$$
Since the value of $\sn^2\left(\frac{\nu}{2}\right)$ is explicitly known in terms of the conserved quantities, the parameter $\lambda_1$ and thus the hyperboloid $\mathcal{H}_1 = \mathcal{C}_1 +\lambda_1\mathcal{C}_3$, hence its coefficients, can be explicitly determined, see Section \ref{The_case_nu}.
With the help of mathematical calculation software it can be verified that $$\iota_{(\nu,\delta)}(x)=\begin{pmatrix}
f(x,\delta)_1\\
f(x,\delta)_2\\
-f(x,\delta)_3\end{pmatrix}$$
which proves the claim for all $0<\nu_i\in\tilde{I}.$\\
Step 1.2. : $\nu_1<0$\\
We follow the exact same procedure as in step 1.1., while setting $\nu_1=\nu-2K_a$. For this case the coefficients are explicitly given in \ref{Sec:The_case_2K_plus_nu}.
Again it can be verified that $$\iota_{(\nu-2K_a),-\delta} = -f(x,\delta) =\tau_{\nu-2K_a}(x)$$ which proves the claim also for $\nu_i<0$.\\ \\
Step 2: $\iota_{(\nu_2,-\sgn(\nu_2)\delta)} = \tau_{\nu_2}^{a}$\\ It follows directly from Step 1 that reversing the signs of $\delta$ gives rise to the maps $\tau_{\nu_2}^a$ defined in (\ref{tau2}).\\ \\
Together we find $$\iota_{(\nu_2,-\sgn(\nu_2) \delta)}\circ \iota_{(\nu_1,\sgn(\nu_1)\delta)} = \tau^{a}_{\nu_2}\circ\tau^{a}_{\nu_1} =f(\cdot,\delta)$$ which proves the statement with Theorem \ref{Thm:lamda_i}.
\end{proof}

\noindent In the following we will give explicit examples of the involutions $\iota_{(\nu_i,\pm\delta)}$ for different values of $\nu_i$ and see for which values these become birational mappings.
The corresponding underplaying computations have been done using mathematica.
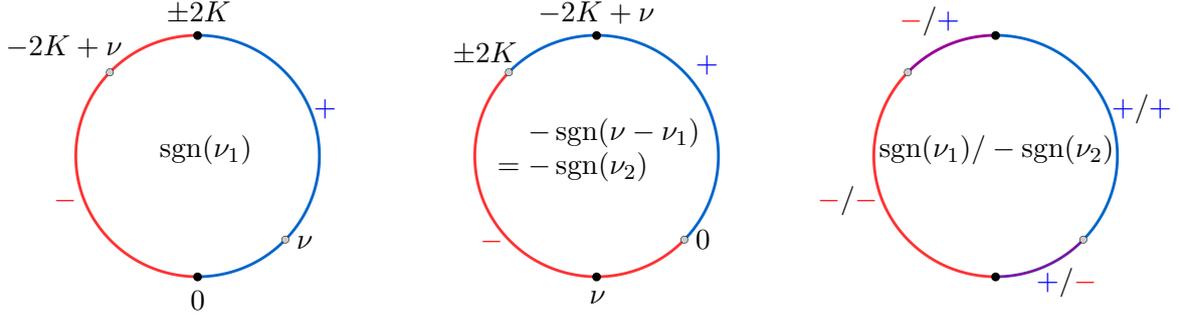
\begin{figure}[t]
\begin{minipage}{0.3\linewidth}
\begin{overpic}[width=\linewidth]
{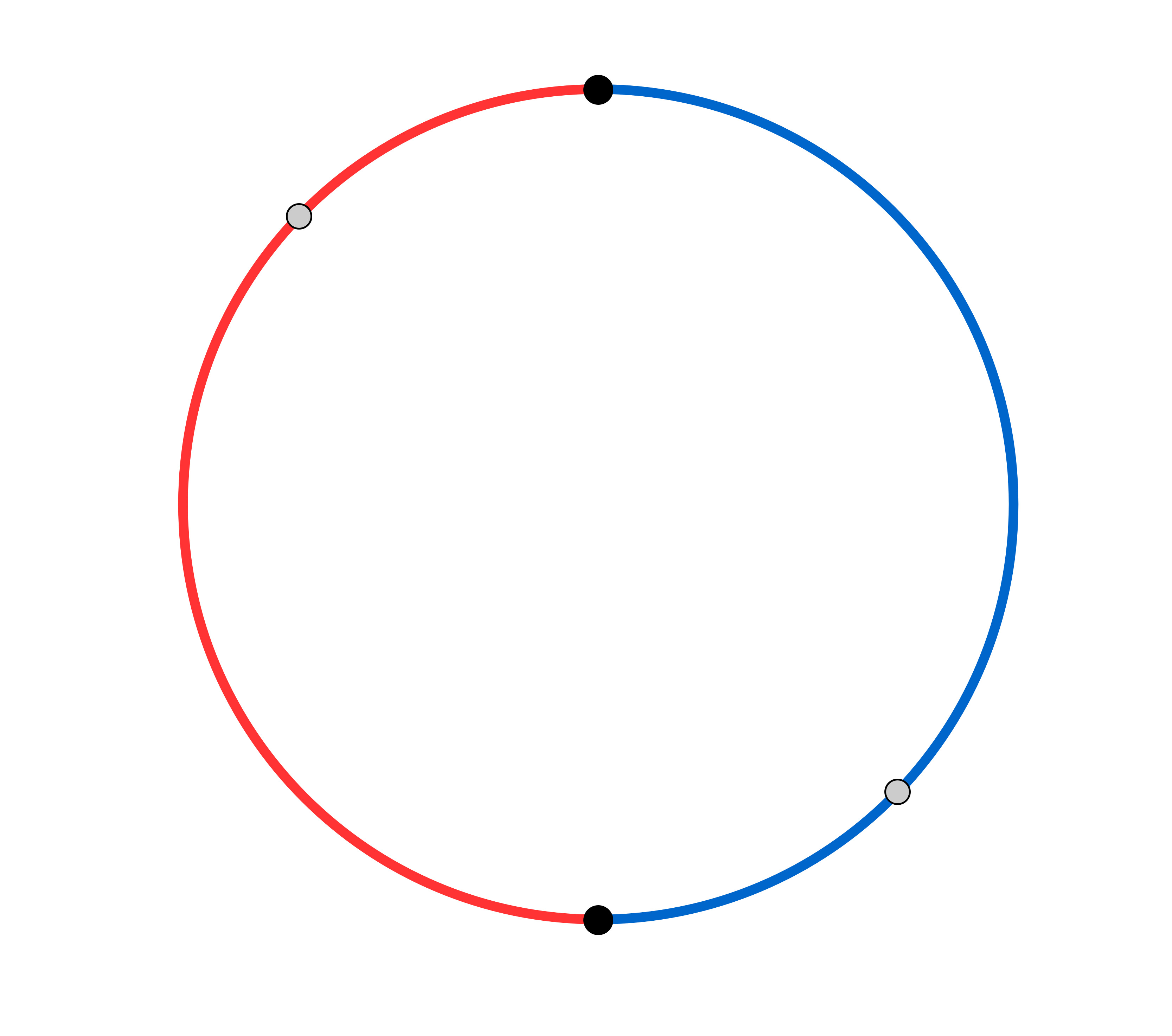}
\put(40,45){$\sgn(\nu_1)$}
\put(85,57){\textcolor{blue}{$+$}}
\put(9,30){\textcolor{red}{$-$}}
\put(49,0){0}
\put(80,18){$\nu$}
\put(-5,75){$-2K+\nu$}
\put(42,85){$\pm2K$}
\end{overpic}
\end{minipage}
\hfill
\begin{minipage}{0.3\linewidth}
\begin{overpic}[width=\linewidth]
{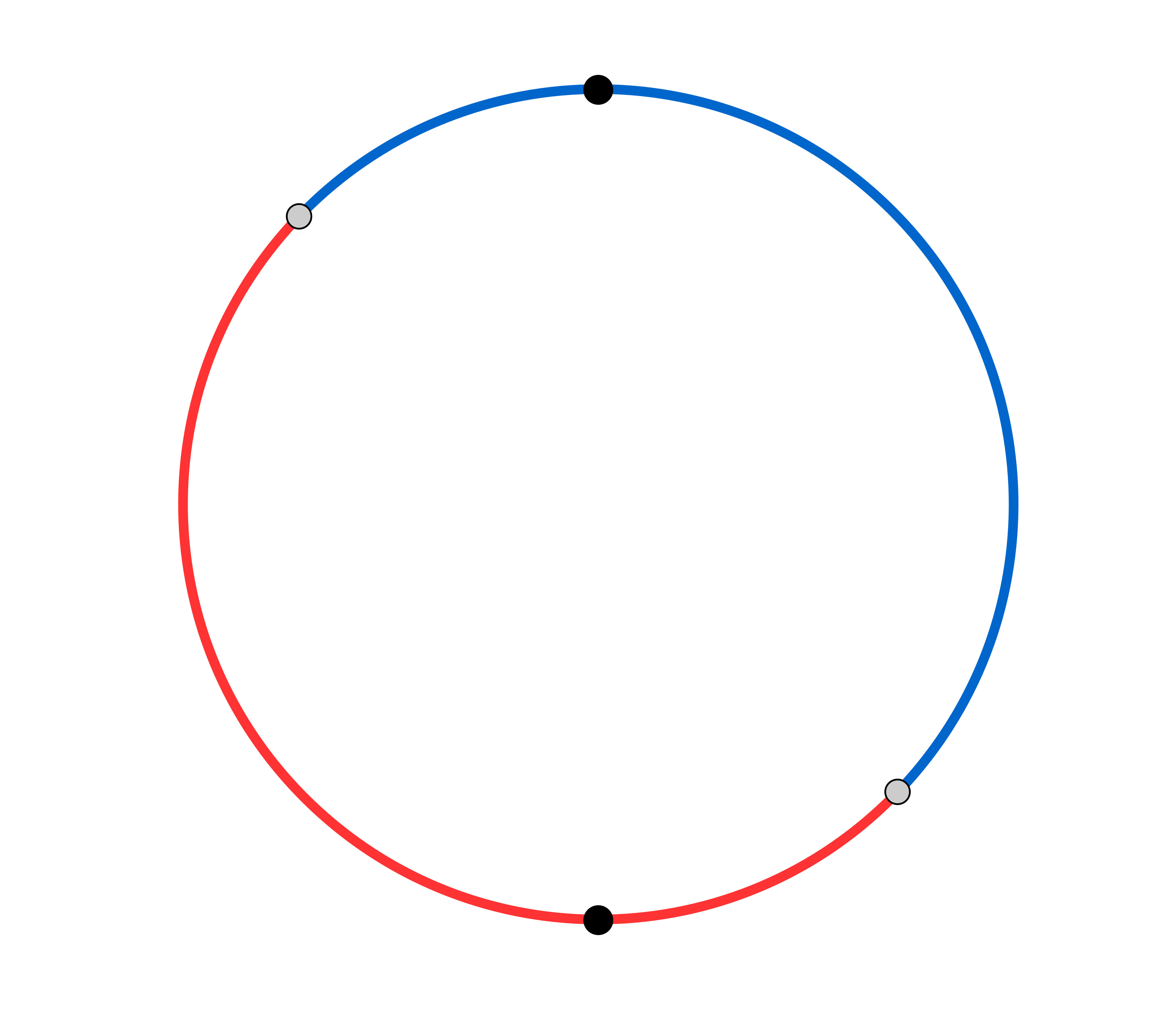}
\put(31,50){$-\sgn(\nu-\nu_1)$}
\put(22,40){$=-\sgn(\nu_2)$}
\put(80,70){\textcolor{blue}{$+$}}
\put(17,18){\textcolor{red}{$-$}}
\put(49,2){$\nu$}
\put(80,18){0}
\put(9,73){$\pm2K$}
\put(34,85){$-2K+\nu$}
\end{overpic}
\end{minipage}
\hfill
\begin{minipage}{0.3\linewidth}
\begin{overpic}[width=\linewidth]
{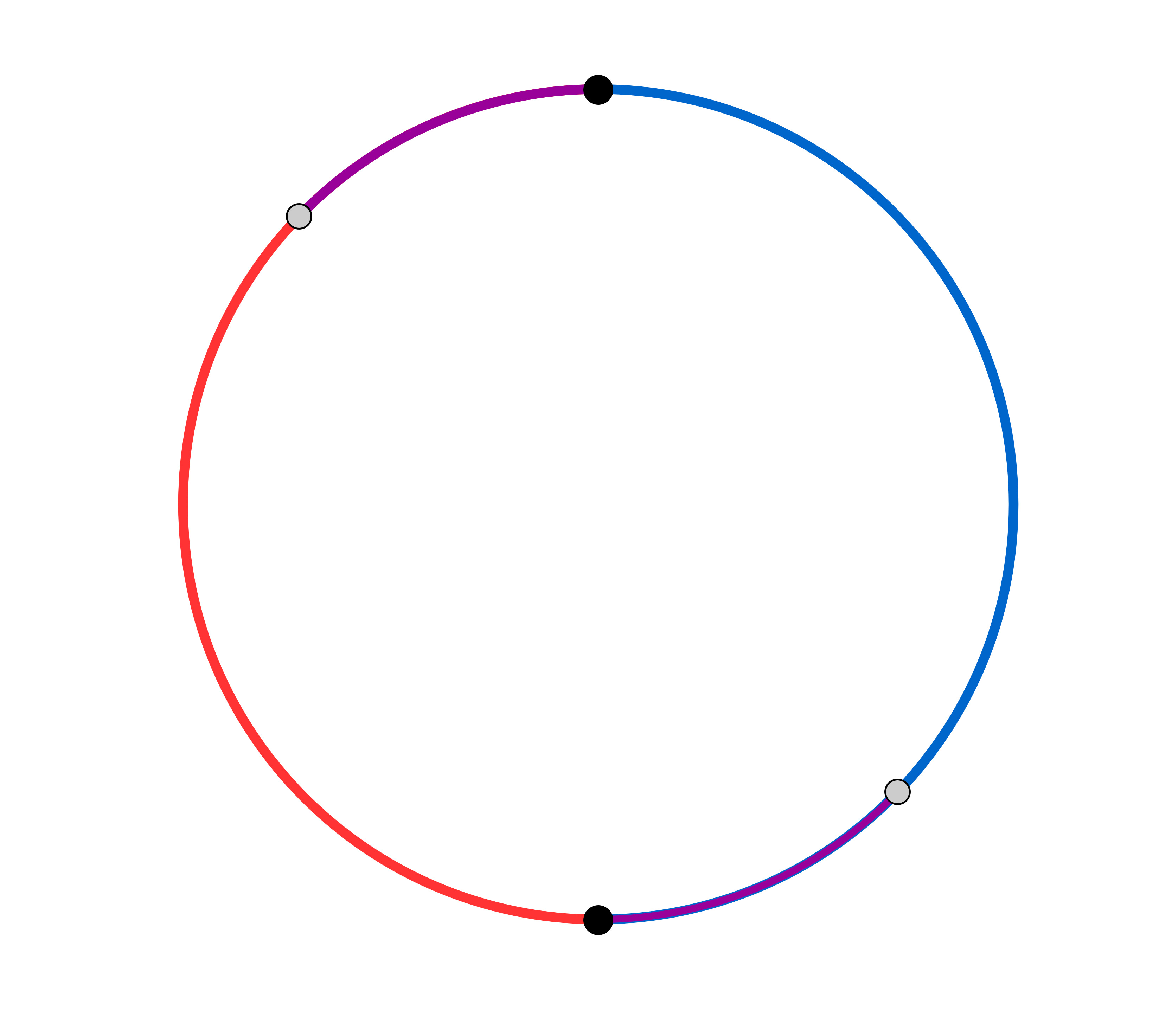}
\put(17,45){$\sgn(\nu_1){/} - \sgn(\nu_2)$}
\put(85,57){\textcolor{blue}{$+$}/\textcolor{blue}{$+$}}
\put(-1,30){\textcolor{red}{$-$}/\textcolor{red}{$-$}}

\put(63,6){\textcolor{blue}{$+$}/\textcolor{red}{$-$}}
\put(23,83){\textcolor{red}{$-$}/\textcolor{blue}{$+$}}
\end{overpic}
\end{minipage}
\caption{For case (a) the dEt is given by the composition $\iota_{(\nu_2,-\sgn(\nu_2) \delta)}\circ \iota_{(\nu_1,\sgn(\nu_1)\delta)}$. The figure can be interpreted in the following sense: Choosing a value for $\nu_1\in\tilde{I}$ corresponds to a unique point on the left circle. The point with the same position in the middle circle corresponds to $\nu_2$, while the colors correspond to the signs, respectively. The right circle then shows these signs that need to be chosen for the parameters $\delta$ for both inversions.}
\label{Fig:sgn}
\end{figure}
\subsection{The square root of the discrete time Euler top}
\label{Sec:sqrt}
Let $\nu_i =\frac{\nu}{2}$.
Then $$\sn^2\left(\frac{\nu_i}{2}\right) = \sn^2\left(\frac{\nu}{4}\right) = \frac{1-\cn\left(\frac{\nu}{2}\right)}{1+\dn\left(\frac{\nu}{2}\right)
}=
\begin{cases}
 \frac{1-\sqrt{F_1}}{1+\sqrt{F_3^{-1}}} &\text{for case (a),}\\
\frac{1-\sqrt{F_3^{-1}}}{1+\sqrt{F_1}} &\text{for case (b).}
\end{cases}$$
due to 
\begin{align*}
&\cn^2\left(\frac{\nu}{2}\right) = F_1, & &\dn^2\left(\frac{\nu}{2}\right)  = F_3^{-1} & &\text{for case (a)} \\
\text{and } 
&\cn^2\left(\frac{\nu}{2}\right) = F_3^{-1}, & &\dn^2\left(\frac{\nu}{2}\right)  = F_1 & &\text{for case (b)}
\end{align*} and half argument identities of $\sn^2$, see \cite[§22.6(iii)]{Nist}.
It follows 
\begin{align*}
\lambda &= -\frac{1-F_1}{1-F_3}\ns^2\left(\frac{\nu}{4}\right) = \frac{1-F_1}{1-F_3}\frac{\sqrt{F_3^{-1}}+1}{\sqrt{F_1}-1}=-\frac{1+\sqrt{F_1}}{\sqrt{F_3}(1-\sqrt{F_3})
} & &\text{ for case (a)}\\
\text{ and }\lambda &= -\frac{1-F_1}{1-F_3}\sn^2\left(\frac{\nu}{4}\right) = \frac{1-F_1}{1-F_3}\frac{\sqrt{F_3^{-1}}-1}{\sqrt{F_1}+1}=\frac{1+\sqrt{F_1}}{\sqrt{F_3}(1+\sqrt{F_3})} & &\text{ for case (b).}
\end{align*} With this the hyperboloid $\mathcal{H}:=\mathcal{C}_1+\lambda \mathcal{C}_3$, i.e., $$\mathcal{H}\colon
\mathcal{A}x_1^2  + \mathcal{B} x_2^2\\ + \mathcal{C} x_3^2 +\mathcal{D}= 0
$$
is defined by  \begin{equation*}
\begin{aligned}
\mathcal{A} &= -\frac{1+\sqrt{F_1}}{\sqrt{F_3}( 1-\sqrt{F_3})} \delta_2\delta_3, & \mathcal{B} &=\frac{1+\sqrt{F_2}}{\sqrt{F_2}(1-\sqrt{F_3})}\delta_1\delta_3,\\
\mathcal{C} &= -F_1\delta_1\delta_2, & \mathcal{D} &= \sqrt{F_1}(1+\sqrt{F_1})(1+\sqrt{F_2})\\
\end{aligned}
\end{equation*} for case (a) and 
\begin{equation*}
\begin{aligned}
\mathcal{A} &= \frac{1+\sqrt{F_1}}{\sqrt{F_3}( 1+\sqrt{F_3})} \delta_2\delta_3, & \mathcal{B} &=-\frac{1-\sqrt{F_2}}{\sqrt{F_2}(1+\sqrt{F_3})}\delta_1\delta_3,\\
\mathcal{C} &= -F_1\delta_1\delta_2 & \mathcal{D}, &= \sqrt{F_1}(1+\sqrt{F_1})(1-\sqrt{F_2})\\
\end{aligned}
\end{equation*}
for case (b). 
With this the involutions $\iota_{(\frac{\nu}{2},\pm\delta)}$ from  (\ref{iotainversion}) simplify to 
\begin{align}
\label{inv_nu_2_case_a}
\iota^{a}_{(\frac{\nu}{2},\pm\delta)}\colon \mathbb{R}^3 \rightarrow \mathbb{R}^3, \qquad x \mapsto \hat{x} \ \text{  with } \ \begin{cases}\hat{x}_1=
\ddfrac{x_1\pm\delta_1x_2x_3}{\sqrt{1-\delta_1\delta_3 x_2^2}\sqrt{1-\delta_1\delta_2 x_3^2}}\\
\hat{x}_2=\ddfrac{x_2\pm\delta_2x_1x_3}{\sqrt{1-\delta_2\delta_3 x_1^2}\sqrt{1-\delta_1\delta_2 x_3^2}}\\
\hat{x}_3=
\ddfrac{-(x_3\pm\delta_3x_1x_2)}{\sqrt{1-\delta_2\delta_3 x_1^2}\sqrt{1-\delta_1\delta_3 x_2^2}}
\end{cases}
\end{align}
for case (a) and 
\begin{align}
\label{inv_nu_2_case_b}
\iota^{b}_{(\frac{\nu}{2},\pm\delta)}\colon \mathbb{R}^3 \rightarrow \mathbb{R}^3, \qquad x \mapsto \hat{x} \ \text{  with } \ \begin{cases}\hat{x}_1=
\ddfrac{-(x_1\mp\delta_1x_2x_3)}{\sqrt{1-\delta_1\delta_3 x_2^2}\sqrt{1-\delta_1\delta_2 x_3^2}}\\
\hat{x}_2=\ddfrac{x_2\mp\delta_2x_1x_3}{\sqrt{1-\delta_2\delta_3 x_1^2}\sqrt{1-\delta_1\delta_2 x_3^2}}\\
\hat{x}_3=
\ddfrac{x_3\mp\delta_3x_1x_2}{\sqrt{1-\delta_2\delta_3 x_1^2}\sqrt{1-\delta_1\delta_3 x_2^2}}
\end{cases}
\end{align}
for case (b).
With Theorem \ref{Thm:real_involution} we find that 
\begin{align}\begin{cases}
\label{iotainversion}
\iota^{a}_{(\frac{\nu}{2},-\delta)} \circ \iota^{a}_{(\frac{\nu}{2},\delta)} = f(\cdot,\delta) & \text{ for case (a)}\\
\iota^{b}_{(\frac{\nu}{2},\delta)}\circ \iota^{b}_{(\frac{\nu}{2},-\delta)} =f(\cdot,\delta) & \text{ for case (b)}\end{cases}
\end{align}
with $f(\cdot,\delta)$ as given in  (\ref{explicit_f}), where for each case the hyperboloids $\mathcal{H}_1$ and $\mathcal{H}_2$ coincide, see Figure \ref{Fig:nu_halbe}.\\ \\
\vspace{-.2cm}
The map \vspace{-.2cm} $$\phi(\cdot,\delta)\colon \mathbb{R}^3 \rightarrow \mathbb{R}^3, \qquad x \mapsto \hat{x}  \ \text{  with } \ \begin{cases}\hat{x}_1=
\ddfrac{x_1+\delta_1x_2x_3}{\sqrt{1-\delta_1\delta_3 x_2^2}\sqrt{1-\delta_1\delta_2 x_3^2}}\\
\hat{x}_2=\ddfrac{x_2+\delta_2x_1x_3}{\sqrt{1-\delta_2\delta_3 x_1^2}\sqrt{1-\delta_1\delta_2 x_3^2}}\\
\hat{x}_3=
\ddfrac{x_3+\delta_3x_1x_2}{\sqrt{1-\delta_2\delta_3 x_1^2}\sqrt{1-\delta_1\delta_3 x_2^2}}
\end{cases}$$ 
can be considered as a discretization of the continuous time Euler top given by (\ref{contin_eq_of_motion}).
On the set where the second iteration is defined iterating it twice constitutes the birational map $f(\cdot,\delta)$ and therefore it is referred to as the \emph{square root} of the HK-type discretization of the Euler top.
This map has been shown to admit remarkable properties considering spherical geometry and related discrete integrable maps, see \cite{petrera2014spherical}.
Similar to  $f(\cdot,\delta)$ the change of the sign of $\delta$ leads to the solutions of the discretization of the Euler top with negative time step.
\begin{figure}[t]
\begin{minipage}{0.48\linewidth}
\begin{overpic}[width=.8\linewidth]
{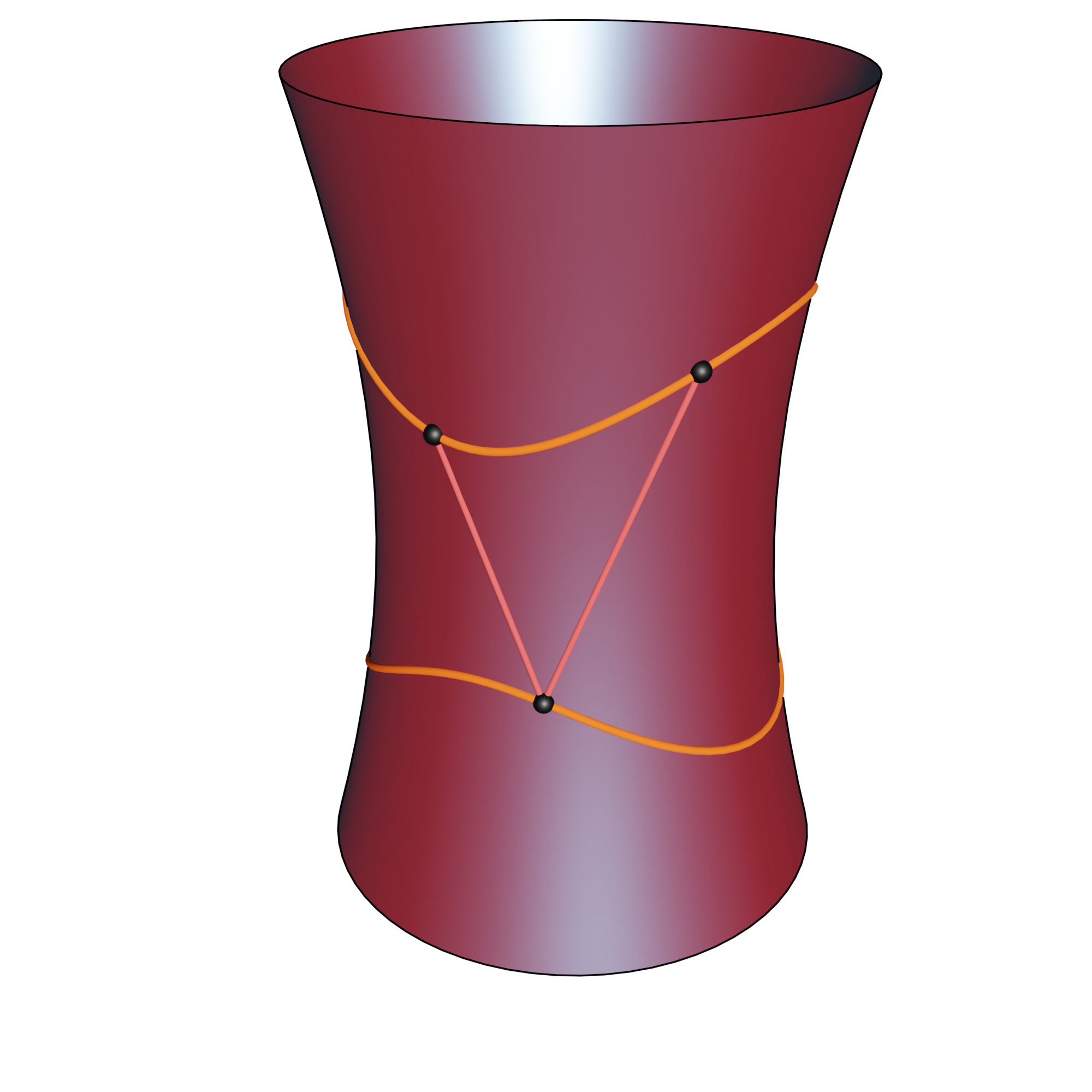}
\end{overpic}
\end{minipage}
\hfill
\begin{minipage}{0.48\linewidth}
\begin{overpic}[width=.8\linewidth]
{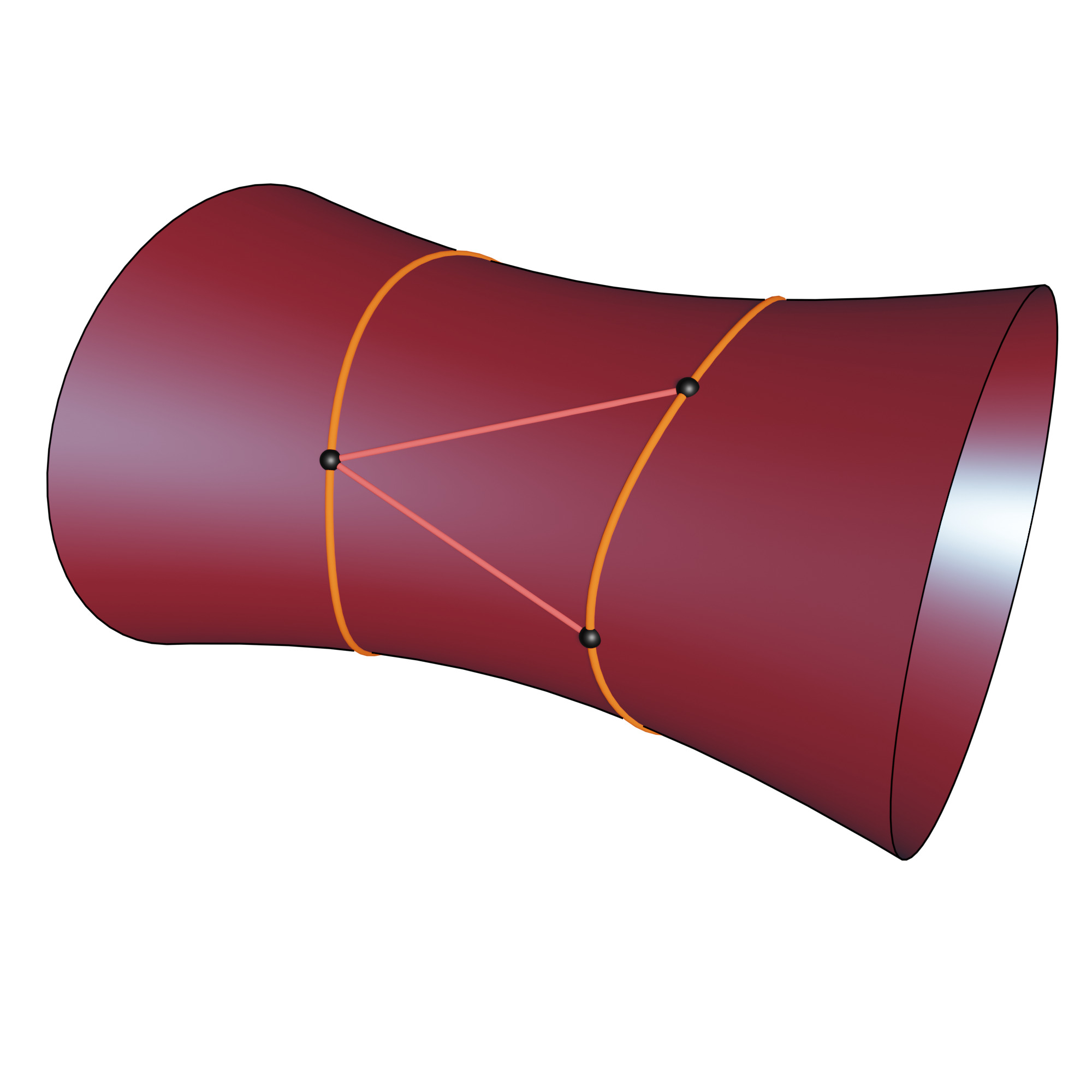} 
\end{overpic}
\end{minipage}
\caption{Case (a) and case (b): Setting $\nu_i=\frac{\nu}{2}$ constitutes the same hyperboloid twice. Both families of generators lie on this hyperboloid and the discrete time Euler top can be described as the composition of two involutions that only differ in the sign of the parameter $\delta$. This case is also shown on the titlepage.}
\label{Fig:nu_halbe}
\end{figure}
\subsection{The involutions $\iota_{(\nu_i,\pm \delta)}$ as birational 3-dimensional maps}
\label{Sec:birational}
We will investigate choices of $\nu_i$ such that these involutions become rational and then investigate the corresponding birational maps.
Let $$f(x,\pm \delta)=\begin{pmatrix}
f(x,\pm \delta)_1\\
f(x,\pm \delta)_2\\
f(x,\pm \delta)_3
\end{pmatrix} $$for the birational map $f$ given in (\ref{explicit_f}).\\
\subsubsection{The case $\nu_i = 0$}
Let $\nu_i =0$.
Then $$\sn^2\left(\frac{\nu_i}{2}\right) = \sn^2\left(0\right) = 0 .$$
It follows 
\begin{align*}
\lambda &= -\frac{1-F_1}{1-F_3}\ns^2\left(\frac{\nu}{2}\right) = \infty & &\text{ for case (a)}\\
\text{ and }\lambda &= -\frac{1-F_1}{1-F_3}\sn^2\left(\frac{\nu}{2}\right) = 0 & &\text{ for case (b).}
\end{align*} With this the quadric $\mathcal{H}=\mathcal{C}_1+\lambda \mathcal{C}_3$ degenerates to the cylinders 
\begin{align*}
\mathcal{H}&=\mathcal{C}_3   \text{ for case (a)}\\
\text{ and }\mathcal{H}&=\mathcal{C}_1 \text{ for case (b).}
\end{align*}
The involutions interchanging two points on $\mathcal{C}_1\cap \mathcal{C}_3$ along generating lines of $\mathcal{H}$ become the trivial birational maps
\begin{align*}
\iota^{a}_{(0,\pm \delta)}\colon \mathbb{R}^3 \rightarrow \mathbb{R}^3, \qquad x \mapsto \hat{x} \ \text{ with } \ \begin{cases}\hat{x}_1=
x_1\\
\hat{x}_2=x_2\\
\hat{x}_3=-
x_3
\end{cases}
\end{align*}
and \begin{align*}
\iota^{b}_{(0,\pm \delta)}\colon \mathbb{R}^3 \rightarrow \mathbb{R}^3, \qquad x \mapsto \hat{x} \ \text{ with } \ \begin{cases}\hat{x}_1=
-x_1\\
\hat{x}_2=x_2\\
\hat{x}_3=
x_3.
\end{cases}
\end{align*}
\subsubsection{The case $\nu_i = \nu$}
\label{The_case_nu}
Let $\nu_i =\nu$.
Then $$\sn^2\left(\frac{\nu_i}{2}\right) = \sn^2\left(\frac{\nu}{2}\right) =
\begin{cases}
 1-F_1 &\text{for case (a),}\\
1-F_3^{-1} &\text{for case (b).}
\end{cases}$$
due to (\ref{case_a_nu}) and (\ref{case_b_nu}).
It follows 
\begin{align*}
\lambda &= -\frac{1-F_1}{1-F_3}\ns^2\left(\frac{\nu}{2}\right) = -\frac{1}{1-F_3} & &\text{ for case (a)}\\
\text{ and }\lambda &= -\frac{1-F_1}{1-F_3}\sn^2\left(\frac{\nu}{2}\right) = F_3^{-1}(1-F_1) & &\text{ for case (b).}
\end{align*} With this the hyperboloid $\mathcal{H}=\mathcal{C}_1+\lambda \mathcal{C}_3$, i.e., $$\mathcal{H}\colon
\mathcal{A}x_1^2  + \mathcal{B} x_2^2\\ + \mathcal{C} x_3^2 +\mathcal{D}= 0
$$
is defined by  \begin{equation*}
\begin{aligned}
\mathcal{A} &= -\frac{1}{1-F_3} \delta_2\delta_3, & \mathcal{B} &= \frac{1}{1-F_3} \delta_1\delta_3,\\
\mathcal{C} &= -F_1\delta_1\delta_2, & \mathcal{D} &= F_1\\
\end{aligned}
\end{equation*} for case (a) and 
\begin{equation*}
\begin{aligned}
\mathcal{A} &= \frac{1-F_1}{F_3} \delta_2\delta_3, & \mathcal{B} &= F_1 \delta_1\delta_3,\\
\mathcal{C} &= -F_1\delta_1\delta_2, & \mathcal{D} &=-\frac{1-F_1}{F_3}\\
\end{aligned}
\end{equation*}
for case (b).
Explicit calculation shows that the involutions $\iota_{(\nu,\pm \delta)}$ from (\ref{iotainversion}) become the birational maps
\begin{align*}
\iota^{a}_{(\nu,\pm \delta)}\colon \mathbb{R}^3 \rightarrow \mathbb{R}^3, \qquad x \mapsto \hat{x} \ \text{ with } \ \begin{cases}\hat{x}_1=
f(x,\pm\delta)_1\\
\hat{x}_2=f(x,\pm\delta)_2\\
\hat{x}_3=-
f(x,\pm\delta)_3
\end{cases}
\end{align*}
and \begin{align*}
\iota^{b}_{(\nu,\pm \delta)}\colon \mathbb{R}^3 \rightarrow \mathbb{R}^3, \qquad x \mapsto \hat{x}\  \text{ with }\ \begin{cases}\hat{x_1}=-
f(x,\pm\delta)_1\\
\hat{x}_2=f(x,\pm\delta)_2\\
\hat{x_3}=
f(x,\pm\delta)_3.
\end{cases}
\end{align*}
The previous two cases have shown that the involutions corresponding to $\nu_i=0$ and $\nu_i=\nu$ are birational maps. With Theorem \ref{Thm:real_involution} we find  $$\begin{cases} \iota^{a}_{(\nu,-\delta)} \circ \iota^{a}_{(0,\pm\delta)} = \iota^{a}_{(0,\pm\delta)} \circ \iota^{a}_{(\nu,\delta)} =f(\cdot,\delta)\quad \text{ for case (a)} \\   \iota^{b}_{(\nu,\delta)} \circ \iota^{b}_{(0,\pm\delta)} = \iota^{b}_{(\nu\pm\delta,)}\circ \iota^{b}_{(\nu,-\delta)} =f(\cdot,\delta)\quad \text{ for case (b)}  \end{cases}$$ which represents the birational map $f(\cdot,\delta)$ as the composition of two birational involutions. The underlying geometric interpretation is shown in Figure \ref{Fig:nu_0}.\\ \\
\begin{figure}[h]
\begin{minipage}{0.48\linewidth}
\begin{overpic}[width=.8\linewidth]
{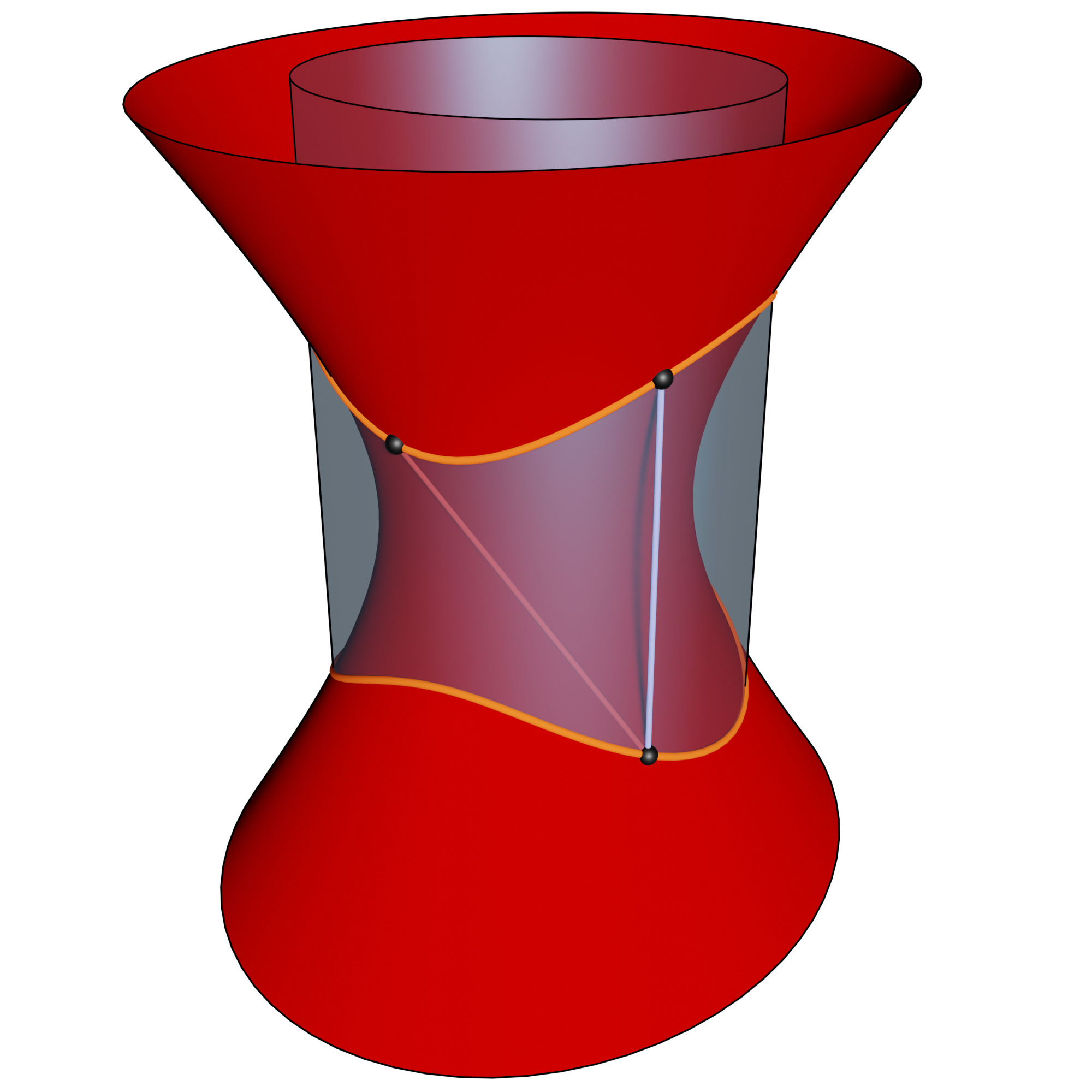}
\end{overpic}
\end{minipage}
\hfill
\begin{minipage}{0.48\linewidth}
\begin{overpic}[width=.8\linewidth]
{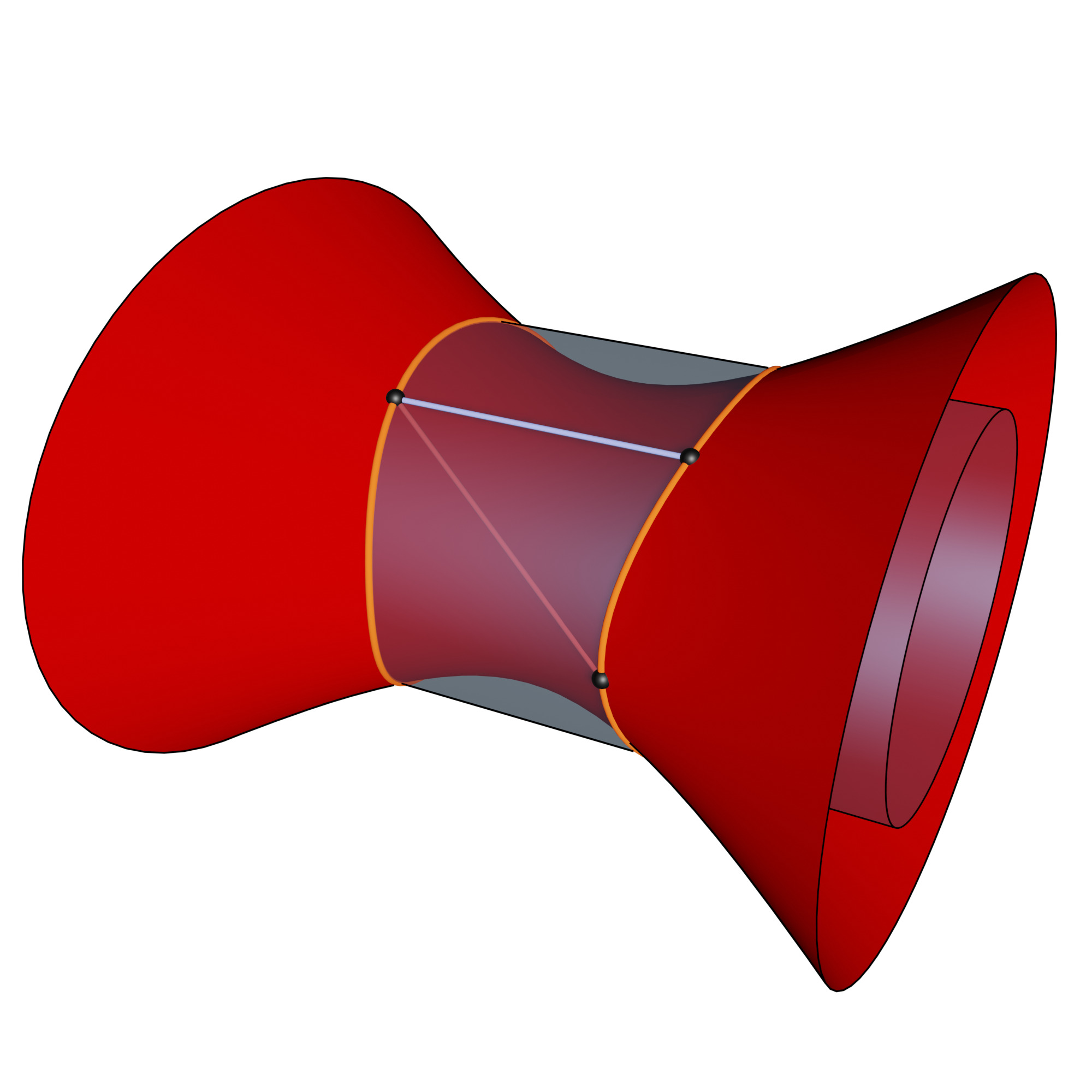} 
\end{overpic}
\end{minipage}
\caption{Case (a) and case (b): Setting $\nu_1=\nu$ and $\nu_2=0$ leads to $\mathcal{H}_2=\mathcal{C}_3$ for case (a) and $\mathcal{H}_2 = \mathcal{C}_1$ for case (b).}
\label{Fig:nu_0}
\end{figure}
\subsubsection{The case $\nu_i  = 2K_\sigma$}
For $\nu_i  = 2K_\sigma$ we have 
$$\sn^2\left(\frac{\nu_i}{2}\right) = \sn^2(K_\sigma) = 1.$$
Thus $$\lambda = -\frac{1-F_1}{1-F_3}$$ for both cases (a) and (b). With this we find that the ruled surface $\mathcal{H}=\mathcal{C}_1+\lambda \mathcal{C}_3$ degenerates to a cone, namely $$\mathcal{H}\colon
\mathcal{A}x_1^2  + \mathcal{B} x_2^2\\ + \mathcal{C} x_3^2 = 0
$$
with  \begin{equation*}
\begin{aligned}
\mathcal{A} &= -(1-F_1)\delta_2\delta_3, & \mathcal{B} &= (1-F_1F_3)\delta_1\delta_3 \ \ \text{ and } &
\mathcal{C}&=-F_1 (1-F_3)\delta_1\delta_2.
\end{aligned}
\end{equation*}
The involutions interchanging two points on $\mathcal{C}_1\cap \mathcal{C}_3$ along generating lines of $\mathcal{H}$ become the trivial birational maps
\begin{align*}
\iota^{\sigma}_{(2K_\sigma,\pm \delta)}\colon \mathbb{R}^3 \rightarrow \mathbb{R}^3, \qquad x \mapsto \hat{x}  \text{ with }  \begin{cases}
\hat{x}_1=-x_1\\
\hat{x}_2=-x_2\\
\hat{x}_3=-x_3
\end{cases}
\end{align*}
\subsubsection{The case $\nu_i=\nu -2K_\sigma$}
\label{Sec:The_case_2K_plus_nu}
Let  $\nu_i=\nu-2K_\sigma$.\\ Then  $$\sn^2\left(\frac{\nu_i}{2}\right)=\sn^2\left(-K_\sigma + \frac{\nu}{2}\right)=
\begin{cases}
\frac{1}{F_2} &\text{for case (a),}\\
F_2 &\text{for case (b).}
\end{cases}$$
due to (\ref{case_a_nu}) and (\ref{case_b_nu}) and addition theorems of Jacobi elliptic functions.\\
It follows 
\begin{align*}
\lambda &=-\frac{1-F_1}{1-F_3}F_2 
\end{align*} for both cases (a) and (b). With this the hyperboloid $\mathcal{H}=\mathcal{C}_1+\lambda \mathcal{C}_3$, i.e.,  $$\mathcal{H}\colon
\mathcal{A}x_1^2  + \mathcal{B} x_2^2\\ + \mathcal{C} x_3^2 +\mathcal{D}= 0
$$
is defined by  \begin{equation*}
\begin{aligned}
\mathcal{A} &=-\frac{1-F_1}{1-F_3}F_2 \delta_2\delta_3, & \mathcal{B}&=  - F_3 \frac{1-F_2}{1-F_3} \delta_1\delta_3,\\
\mathcal{C} &= -F_1\delta_1\delta_2, & \mathcal{D} &=-(1-F_1)(1-F_2).
\end{aligned}
\end{equation*}\\
Explicit calculation shows that the involutions $\iota^{\sigma}_{(\nu-2K_\sigma,\pm \delta)}$ from  (\ref{iotainversion}) become the birational maps
\begin{align*}
&\iota^{a}_{(\nu-2K_a,\pm \delta)}\colon \mathbb{R}^3 \rightarrow \mathbb{R}^3, \qquad x \mapsto \hat{x}=-f(x,\mp\delta)\\
\text{and} \qquad &\iota^{b}_{(\nu-2K_b,\pm \delta)}\colon \mathbb{R}^3 \rightarrow \mathbb{R}^3, \qquad x \mapsto \hat{x}=-f(x,\pm\delta).\\
\end{align*}
Again we can combine the previous two cases $\nu_i = 2K_\sigma$ and $\nu_i=\nu-2K_\sigma$
and with Theorem \ref{Thm:real_involution} we find that $$\begin{cases} \iota^{a}_{(2K_a,\pm\delta)} \circ \iota^{a}_{(\nu-2K_a,-\delta)} = \iota^{a}_{(\nu-2K_a,\delta)} \circ \iota^{a}_{(2K_a, \pm \delta)} =f(\cdot,\delta)\quad \text{ for case (a)} \\  \iota^{b}_{(2K_b,\pm\delta)} \circ \iota^{b}_{(\nu-2K_b,\delta)} = \iota^{b}_{(\nu-2K_b,-\delta)} \circ \iota^{b}_{(2K_b, \pm \delta)} =f(\cdot,\delta)\quad \text{ for case (b)}  \end{cases}$$ which represents the birational map $f(\cdot,\delta)$ as the composition of another two birational involutions.\newpage
\section{Concluding remarks}
We have shown that the solutions of the  Hirota-Kimura type discretization of the Euler top can be described as the composition of two involutions. These can be investigated in complex or real space and we found examples such that they can be expressed as birational maps in $\mathbb{R}^3$. In particular these examples are those where at least one of the involutions coincides with the birational map, given by the dEt, up to signs and factor. Further examples are not known yet.\\ \\  Beyond the description of the dEt as the composition of two involutions, the geometric background opens new questions and fields of study corresponding the discrete time Euler top. For example the geometry introduces a close correspondence to QRT maps and the theorem of Poncelet, which might be of great interest to study deeper. A correspondence between the geometry and QRT maps can be found in \cite{bobenko2020noneuclidean}.
\newpage
\bibliographystyle{plain}
\bibliography{Masterarbeit}
\end{document}